\documentclass[a4paper,UKenglish,cleveref,pdfa]{lipics-v2021}

\pdfoutput=1
\hideLIPIcs
\nolinenumbers

\graphicspath{{./figures/}}

\usepackage{preamble}

\title{Temporal Path Covers: Dilworth Properties and Parameterized Complexity}

\author{Lapo Cioni}
{Department of Statistic, Computer Science and Applications, Università degli Studi di Firenze, Firenze, Italy}
{lapo.cioni@unifi.it}
{https://orcid.org/0000-0002-4605-8473}{}

\author{Sotiris Kanellopoulos}
{National Technical University of Athens and Archimedes, Athena Research Center, Greece \and \url{https://skanellopoulos.github.io/}}
{s.kanellopoulos@athenarc.gr}
{https://orcid.org/0009-0006-2999-0580}{}

\author{Edouard Nemery}
{Universit\'{e} Paris Dauphine -- PSL, CNRS UMR7243, LAMSADE, Paris, France}
{edouard.nemery@dauphine.psl.eu}
{https://orcid.org/0009-0007-6977-9330}{}

\author{Aris Pagourtzis}
{National Technical University of Athens and Archimedes, Athena Research Center, Greece}
{pagour@cs.ntua.gr}
{https://orcid.org/0000-0002-6220-3722}{}

\author{Christos Pergaminelis}
{National Technical University of Athens and Archimedes, Athena Research Center, Greece}
{chr.pergaminelis@athenarc.gr}
{https://orcid.org/0009-0009-8981-3676}{}

\author{Manolis Vasilakis}
{University of Southern Denmark, Odense, Denmark
\and \url{https://manolisvasilakis.github.io}}
{vasilakis@imada.sdu.dk}
{https://orcid.org/0000-0001-6505-2977}{}

\authorrunning{L. Cioni, S. Kanellopoulos, E. Nemery, A. Pagourtzis, C. Pergaminelis, and M. Vasilakis} 

\Copyright{Lapo Cioni, Sotiris Kanellopoulos, Edouard Nemery, Aris Pagourtzis, Christos Pergaminelis, and Manolis Vasilakis} 

\ccsdesc[100]{Theory of computation~Problems, reductions and completeness}
\ccsdesc[500]{Theory of computation~Parameterized complexity and exact algorithms}

\keywords{Temporal graphs, Path cover, Dilworth's theorem, Lov{\'a}sz number, Parameterized complexity} 

\category{} 

\relatedversion{} 

\EventEditors{John Q. Open and Joan R. Access}
\EventNoEds{2}
\EventLongTitle{42nd Conference on Very Important Topics (CVIT 2016)}
\EventShortTitle{CVIT 2016}
\EventAcronym{CVIT}
\EventYear{2016}
\EventDate{December 24--27, 2016}
\EventLocation{Little Whinging, United Kingdom}
\EventLogo{}
\SeriesVolume{42}
\ArticleNo{23}

\begin{document}

\maketitle

\begin{abstract}
The \textsc{Temporal Path Cover} (TPC) and \textsc{Temporally Disjoint Path Cover} (TDPC) problems were introduced by [Chakraborty, Dailly, Foucaud, Klasing, MFCS '24] as the natural temporal analogues of \textsc{Path Cover} and \textsc{Path Partition}, and were shown to be \NP-hard on temporal DAGs. All tractable graph classes for T(D)PC established in that paper satisfy a temporal \emph{Dilworth property}, namely that the size of a respective minimum temporal path cover is equal to the size of a maximum temporal antichain. This raises a natural question: is T(D)PC guaranteed to be polynomial-time solvable under the promise that the respective Dilworth property holds? For both problems, we answer this question in the affirmative for their decision version and in the negative for their search version (under $\P \neq \NP$). As an intermediate step we establish that, under the respective Dilworth promise, the shared size of a minimum cover and a maximum antichain is exactly equal to the Lov{\'a}sz number of the connectivity graph.

In another direction, we present parameterized algorithms and hardness results for TPC and TDPC. Our main result is that TPC is \W$[1]$-hard parameterized by the \emph{deletion distance to linear forest} even for temporal graphs with two time-steps, answering in the negative an open question by Chakraborty et al. about whether an {\XP} algorithm parameterized by the treewidth plus the number of time-steps can be improved to an \FPT\ one. On the other hand, we prove that an {\FPT} algorithm does exist if the vertex cover number is used as parameter instead of the treewidth in the above parameterization. We complement this with a proof that including the number of time-steps in the parameter is necessary to yield tractability, as,
otherwise, both TPC and TDPC remain \NP-hard even for constant vertex cover size.
Along the way, we establish several additional para-\NP-hardness results involving structural parameters such as the pathwidth and the maximum degree of the underlying graph.

\end{abstract}

\newpage






\section{Introduction}

Temporal graphs extend the concept of static graphs to include time-dependent information, with transitions through edges being possible only if the respective edge is active in the current time-step. As such, temporal graphs naturally model scenarios such as road networks, bus or train transits, and flight routes. Recent research has focused on pathfinding~\cite{temporal_beer_path_IWOCA,Path_Wu}, temporal walks and exploration~\cite{Himmel_optimal_walk,Erlebach_exploration}, temporal flows and cuts~\cite{Akrida_temporal_flow} and reachability~\cite{Deligkas_temporal_parameterized,Deligkas_shifting,Deligkas_reachability_IJCAI,DBLP:conf/isaac/Doring25}, among others. 
For a comprehensive overview on temporal graphs, we refer to~\cite{Holme_Temporal_graphs,Wang_Temporal_graphs}.

In this work we study the temporal analogues of \PC\ and \PPart, where the former problem asks whether all vertices of a (static) graph can be covered by $k$ paths, while the latter additionally demands that the paths are vertex-disjoint (cf.~\cite{Dusart_dpc,Foucaud_static_caldam,Hung_pc_dpc}).
The temporal analogues of these problems, namely \TPC\ (\TPCshort) and \TDPC\ (\TDPCshort), were introduced by Chakraborty, Dailly, Foucaud and Klasing~\cite{temporal_path_cover_Dailly_MFCS} in MFCS 2024. Thus, \TPCshort\ asks whether all vertices of a temporal graph can be covered by $k$ \emph{temporal} paths, whereas \TDPCshort\ additionally demands that the paths are \emph{temporally disjoint}, meaning that two paths cannot occupy the same vertex \emph{at the same time unit}.

\TPCshort\ and \TDPCshort\ are naturally connected with the concept of temporal \emph{antichains}, that is, sets of pairwise temporally unreachable vertices. It is easy to see that the cardinality of any temporal path cover (or temporally disjoint path cover) is at least as large as the size of a maximum temporal antichain, since each vertex in an antichain has to be covered by a different path. Chakraborty et al.~\cite{temporal_path_cover_Dailly_MFCS} established tractability and hardness results for various temporal graph classes and, interestingly, found that classes admitting polynomial-time algorithms for \TPCshort\ (resp. \TDPCshort) seem to satisfy a \emph{Dilworth property}, i.e., the cardinality of a minimum temporal path cover (resp. temporally disjoint path cover) is guaranteed to be exactly equal to the size of a maximum temporal antichain. On the other hand, classes established to be \NP-hard (such as temporal DAGs) were shown to violate the respective property. This intriguing observation raises the following question: does the promise of a Dilworth property always guarantee tractability for \TPCshort\ and \TDPCshort?


In the same paper, an interesting contrast is shown between \TPCshort\ and \TDPCshort. On the one hand, the former is tractable in \emph{temporal oriented trees}, while the latter is \NP-hard on the same graph class.
On the other hand, the latter is shown to admit an {\FPT} algorithm parameterized by the treewidth of the underlying graph plus the number of time-steps, while the former is only shown to admit an {\XP} algorithm in the same regime; Chakraborty et al.~\cite{temporal_path_cover_Dailly_MFCS} explicitly ask whether \TPCshort\ also admits an {\FPT} algorithm.
A negative resolution of that question would be of particular interest, since it would highlight the complex relationship between the two problems, as there would be both cases in which \TPCshort\ is presumably easier than \TDPCshort, and cases in which the converse holds.

In this work, we settle both the aforementioned questions by Chakraborty et al.~\cite{temporal_path_cover_Dailly_MFCS}. Our first main result establishes that, under the Dilworth promise, the size of a minimum temporal path cover is exactly equal to the \emph{Lov{\'a}sz number} of the connectivity graph, thus also guaranteeing tractability by means of the ellipsoid method. Surprisingly, however, we show that this tractability only holds for the decision version of the problem (computing the cardinality of a minimum temporal path cover); indeed, a polynomial-time algorithm for the search version of \TPCshort\ (computing a minimum temporal path cover) would imply $\P=\NP$. Thus, \TPCshort\ under the Dilworth promise is one of the rare computational problems with easy decision and presumably hard search version. All of the aforementioned results also hold for \TDPCshort.

Our second main result is that \TPCshort\ is \W$[1]$-hard parameterized by the \emph{deletion distance to linear forest} even for temporal graphs with two time-steps, which rules out the existence of an {\FPT} algorithm parameterized by the treewidth plus the number of time-steps under standard complexity assumptions. This answers the respective question from~\cite{temporal_path_cover_Dailly_MFCS} in the negative and establishes a contrast with \TDPCshort, which admits such an {\FPT} algorithm. 
We complement this result with various parameterized complexity results for \TPCshort\ and \TDPCshort.
See \cref{sec:contributions} for a comprehensive list of all our results.

\subsection{Related work}\label{sec:related_work}

\subparagraph{Path Covers in static graphs.}
The starting point for this work is Dilworth's theorem, which states that
in every partially ordered set the minimum number of chains needed to cover all elements is equal to the maximum size of an antichain~\cite{Dilworth_original}.  Equivalently, in a directed acyclic graph (DAG), the minimum number of directed paths needed to cover all vertices is equal to the maximum size of a set of pairwise unreachable vertices. Fulkerson~\cite{Fulkerson1956} gave an algorithmic proof of this duality, showing a polynomial-time algorithm for the optimization version of the two problems. 
Improving and generalizing algorithms for path covers and chain covers remains a very active direction, with recent work including
almost linear-time algorithms for \textsc{Minimum Chain Cover} and parameterized algorithms for \textsc{Minimum Path Cover} in DAGs~\cite{Caceres_icalp_chains,Caceres_soda_chains}.

Beyond DAGs, Foucaud, Majumder, M\"{o}mke and Roshany-Tabrizi~\cite{Foucaud_static_caldam} gave polynomial-time algorithms for \PC\ and \PPart\ on trees and parameterized algorithms for graphs of bounded treewidth. Specifically, they gave an {\XP} algorithm for \PC\ on graphs of treewidth $t$ with running time $n^{t^{O(t)}}$, and an {\FPT} algorithm for the corresponding \PPart\ problem, with running time $2^{O(t)}n^{O(1)}$.

\subparagraph{Temporal(ly) (Disjoint) Path Covers.}
Chakraborty et al.~\cite{temporal_path_cover_Dailly_MFCS} introduced \TPCshort\ and \TDPCshort\ as the temporal analogues of \textsc{Path Cover} and \textsc{Path Partition} respectively. The authors proved that both \TPCshort\ and \TDPCshort\ remain \NP-hard even on very restricted temporal DAGs with only two time-steps, contrasting the aforementioned results for static graphs~\cite{Fulkerson1956}. 
On the positive side, they showed that \TPCshort\ admits an algorithm running in time $O(\pi n^2+n^3)$ on temporal oriented trees, where $n$ denotes the number of vertices and $\pi$ the maximum number of time labels per edge. However, this does not hold for TDPC, which remains \NP-hard even on temporal oriented trees. 
Nevertheless, they identified two other tractable subclasses: both \TPCshort\ and \TDPCshort\ can be solved in $O(\pi n)$ time on temporal oriented lines and in $O(\pi n^2)$ time on temporal rooted directed trees. 
Classes that admit polynomial-time algorithms for \TPCshort\ (resp. \TDPCshort) were shown to satisfy a Dilworth (resp. TD-Dilworth\footnote{``TD'' here stands for \emph{temporally disjoint}.}) property, namely that the cardinality of the minimum temporal path cover (resp. temporally disjoint path cover) is guaranteed to be exactly equal to the size of the maximum temporal antichain. On the other hand, \NP-hard classes were shown to violate the respective property. We refer to \cref{table:dilworth_classes} for an overview of these classes. 
The authors additionally studied the parameterized complexity of the two problems on general temporal digraphs, and in particular parameterized by the
treewidth of the underlying undirected graph together with the number of time-steps.
For \TDPCshort, they gave an {\FPT} algorithm, while for \TPCshort\ they gave an {\XP} algorithm.
Whether the latter also admits an {\FPT} algorithm was left open.

\newcommand{\lightrule}{\arrayrulecolor{black!30}\hline\arrayrulecolor{black}}
{
\begin{table}[ht]
\centering
\begin{tabular}{|c|c!{\color{black!30}\vrule}c|}
\hline
\textbf{Temporal graph class} & \textbf{Dilworth property} & \textbf{TD-Dilworth property} \\
\hline
Temporal DAGs & \xmark & \xmark \\
\lightrule
Temporal oriented trees & \cmark & \xmark \\
\lightrule
Temporal rooted directed trees & \cmark & \cmark \\
\lightrule
Temporal oriented lines & \cmark & \cmark \\
\hline
\end{tabular}
\caption{The temporal graph classes studied by Chakraborty et al.~\cite{temporal_path_cover_Dailly_MFCS}. All classes in this table that satisfy the Dilworth property admit polynomial-time algorithms for \TPCshort, while all classes that violate it are \NP-hard. Symmetrically, the same holds for the TD-Dilworth property and \TDPCshort. For the precise definitions of these classes, see \cref{subsec:prelims_temporal}.}
\label{table:dilworth_classes}
\end{table}
}

\subparagraph{Further Related Work.}
The disjointness condition in \TDPCshort\ is closely related to the concept of  \emph{temporally disjoint paths} introduced by Klobas, Mertzios, Molter, Niedermeier, and Zschoche~\cite{Mertzios_disjoint_paths}.
In the respective problem, one is given terminal pairs and asks whether they can be connected by pairwise temporally disjoint temporal paths. Kunz, Molter, and Zehavi~\cite{Kunz_tdpc_ijcai} further studied temporally disjoint paths and walks under structural restrictions on the underlying graph. In a complementary direction to the vertex covering perspective of \TPCshort\ and \TDPCshort, Deligkas, Döring, Eiben, Goldsmith, Skretas, and Tennigkeit~\cite{Deligkas_temporal_edge_cover} introduced the \textsc{Temporal Exact Edge Cover} problem, asking whether all temporal edges can be covered exactly once by a given amount of temporal paths. 

\subsection{Our Contributions}\label{sec:contributions}

The paper is structured as follows. \cref{sec:prelims} contains problem definitions and relevant notation, while \cref{sec:introductory_results} contains some simple hardness results that are necessary to understand the context of this paper, namely \NP-hardness for \TAC\ and for verifying the Dilworth property.

In \cref{sec:Dilworth_Lovasz} we prove that the Lov{\'a}sz number of the connectivity graph is sandwiched between the size of a maximum temporal antichain and the size of a minimum temporal path cover. As such, when the Dilworth property is promised, this number is exactly equal to both values. Since the Lov{\'a}sz number of a graph can be computed with arbitrary accuracy in polynomial time by semidefinite programming~\cite{Lovasz_perfect}, and the number in this case is guaranteed to be an integer due to the aforementioned sandwich result, it follows that the shared size of a minimum cover and a maximum antichain can be computed exactly in polynomial time when the Dilworth property holds. This partially answers in the positive the relevant question from~\cite{temporal_path_cover_Dailly_MFCS}. However, it still remains to investigate whether a minimum cover can also be recovered in polynomial time. Surprisingly, we answer this in the negative, proving that a polynomial-time algorithm for the search version of \TPCshort\ with the Dilworth property would imply $\P=\NP$.\footnote{Note that this intractability result is not technically \NP-hardness, despite ruling out polynomial-time algorithms under $\P \neq \NP$. Such conditional intractability results are known to arise for search problems with promises that are \NP-hard to verify (cf. Gottlob~\cite{Gottlob_min_constraint_networks_search_problem}), like the Dilworth promise in our case.} Thus, we establish a sharp contrast between the decision and search versions of \TPCshort\ under the Dilworth property.
All of the aforementioned results also hold for \TDPCshort\ with the TD-Dilworth property.

In \cref{sec:tpc-two-paths} we investigate whether our aforementioned tractability result for computing the value of \TPCshort\ and \TDPCshort\ under the (TD-)Dilworth promise can be generalized to a setting where only a relaxed promise is given, i.e., that the difference in size between a minimum cover and a maximum antichain is small. Surprisingly, we prove that both problems become \NP-hard even under relaxed promises that demand that the respective difference is at most $1$. We achieve this by utilizing the fact that \textsc{Hamiltonian Cycle} remains \NP-complete even when a Hamiltonian path is known~\cite{GareyJohnson,stoc/PapadimitriouS76}.

In \cref{sec:TPC_parameterized} we prove that \TPCshort\ is \W$[1]$-hard parameterized by the \emph{deletion distance to linear forest} (\cref{def:ddlf}) even for temporal graphs with two time-steps, thus ruling out an {\FPT} algorithm parameterized by treewidth plus number of time-steps. We achieve this through a reduction from \RUBP, a variant of \textsc{Unary Bin Packing} which is also known to be \W$[1]$-hard by the number of bins~\cite{mfcs/HanakaLVY24}.  This answers in the negative the respective question by Chakraborty et al.~\cite{temporal_path_cover_Dailly_MFCS} and shows a contrast between \TPCshort\ and \TDPCshort, since the latter is known to admit an {\FPT} algorithm under the same parameterization~\cite{temporal_path_cover_Dailly_MFCS}. We complement this hardness result with an {\FPT} algorithm for \TPCshort\ in case we replace treewidth with vertex cover number in the above parameterization, by reducing the problem to an ILP of bounded dimension and employing Lenstra's algorithm~\cite{Lenstra_ILP}.

In \cref{sec:para-NP-hardness} we present a plethora of para-{\NP}-hardness results for \TPCshort\ and \TDPCshort, as well as for a restricted variant of the static \PC\ problem in which the length of each path in the cover is bounded by a given value. To this end, we modify a recent reduction from \NMTS\ to \textsc{Telephone Broadcast} by Egami et al.~\cite{mfcs/EgamiGHKLMNOV025}.
Specifically, we prove that \TPCshort, \TDPCshort, and \RLPC\ remain \NP-hard even for constant \emph{deletion distance to linear forest}, and via a slight adaptation similarly obtain \NP-hardness for constant maximum degree and pathwidth. Note that this contrasts the status of the (unrestricted) \PC\ problem, which is known to admit an \XP\ algorithm parameterized by treewidth~\cite{Foucaud_static_caldam}. Our result additionally highlights that \TPCshort\ becomes \NP-hard when the \emph{feedback vertex set number} is $1$, while the problem is known to be tractable when it is $0$, i.e., for temporal oriented trees~\cite{temporal_path_cover_Dailly_MFCS}. Lastly, we prove that dropping the number of time-steps from the parameterization of our \FPT\ algorithm in \cref{sec:TPC_parameterized} is impossible, since \TPCshort\ and \TDPCshort\ both remain \NP-hard even on temporal DAGs with constant vertex cover size.

\section{Preliminaries}\label{sec:prelims}

Throughout the paper, we use the notations $[n]=\{1,\ldots,n\}$ and $[m,n]=\{m,\ldots,n\}$, for $m,n\in \mathbb{N}$ with $m\leq n$. 
The proofs of statements marked with $(\appsymb)$ are deferred to the appendix due to space constraints.

\subsection{Temporal graphs and notation}\label{subsec:prelims_temporal}

A temporal digraph is a time-indexed digraph $G=(V,E_1,\ldots,E_{t_{\max}})$, where the set of time-steps is $\{1,\ldots,t_{\max}\}$, and $E_i$ represents the edge-set that contains exactly the directed edges that are available at time-step $i$. The \emph{underlying digraph} of the temporal graph is denoted by $U_{G}=(V,E)$, where $E=\bigcup_{i=1}^{t_{\max}} E_i$.
An equivalent representation for temporal graphs can be obtained by assigning to every edge the set of time-steps at which it is active. More formally, one can use an edge-labeling function $\lambda \colon E(G) \to 2^{[t_{\max}]},$ where, for a directed edge $e\in E(G)$, the label set $\lambda(e)\subseteq [t_{\max}]$ is the set of time-steps in which $e$ is active. Under this representation, we denote the temporal digraph as $G=(U_{G},\lambda).$
In the context of this paper, all temporal graphs have unit \emph{traversal times} for all their edges, i.e., if an edge $(u,v)$ is active at time unit $t$, then it is possible for a temporal path to start from $u$ at time $t$ and reach~$v$ at time $t+1$.
Throughout this paper, we use the notations $n=|V|$ and $m=|E|$. Additionally, we use $\pi$ to denote the maximum number of time-steps assigned to any edge, that is, $\pi=\max_{e\in E(U_{G})}|\lambda(e)|$. Lastly, $\underline{U_G}$ denotes the underlying \emph{undirected} graph of a temporal digraph $G$.

A temporal digraph $G=(U_G,\lambda)$ is:
\begin{itemize}
    \item a \emph{temporal DAG}, if $U_G$ is a directed acyclic graph (DAG).
    \item a \emph{temporal oriented tree} if $\underline{U_G}$ is a tree;
    \item a \emph{temporal rooted directed tree} if it is a temporal oriented tree and $U_{G}$ has exactly one vertex of in-degree zero;
    \item a \emph{temporal oriented line} if $\underline{U_G}$ is a path.
\end{itemize}

In a temporal digraph, a strict temporal directed path (or simply \emph{temporal path}) is a sequence $(v_1,v_2,t_1),(v_2,v_3,t_2),\ldots,(v_{k-1},v_k,t_{k-1})$ such that for any $i,j$ with $1 \leq i < j \leq k$, it holds that $v_i \neq v_j$, and for any $i$ with $1 \leq i \leq k-1$, it holds that $t_i < t_{i+1}$ and the edge $(v_i, v_{i+1})$ is active at time-step $t_i$. For a temporal path  $P=(v_1,v_2,t_1),\ldots,(v_{k-1},v_k,t_{k-1}),$ we denote by $V(P)$ the set $ \bigcup_{i=1}^{k}\{v_i\}$ and by $E(P)$ the set $\bigcup_{i=1}^{k-1}\{(v_i, v_{i+1})\}.$
We say that a temporal path $P$ \emph{covers} the vertices in $V(P)$.

We say that a temporal path $P=(v_1,v_2,t_1),\ldots,(v_{k-1},v_k,t_{k-1})$ \emph{occupies} vertex $v_i$, $i\in [2,k-1]$, throughout the time interval $[t_{i-1}+1,\ t_i]$, and vertices $v_1,v_k$ at time-steps $t_1,\ t_{k-1}+1$ respectively.
Two temporal paths are said to be \emph{temporally disjoint} if they do not occupy a vertex at the same time-step (cf.~\cite{temporal_path_cover_Dailly_MFCS,Mertzios_disjoint_paths}). We say that two vertices $u,v \in V$ are \emph{temporally connected} if there exists a temporal path from $u$ to $v$ or one from $v$ to $u$.

\subsection{Temporal path cover and antichain problems}\label{subsec:prelims_problems}

\begin{definition}[Temporal path cover]\label{def:TPC}
    Let $G$ be a temporal digraph. A \emph{temporal path cover} of $G$ is a collection $\mathcal{C}$ of temporal paths such that every vertex of $G$ is covered by at least one path in $\mathcal{C}$, that is, $\bigcup_{P\in\mathcal{C}}V(P) = V$.
\end{definition}

\begin{definition}[Temporally disjoint path cover]\label{def:TDPC}
    Let $G$ be a temporal digraph. A \emph{temporally disjoint path cover} of $G$ is a temporal path cover $\mathcal{C}$ such that any two distinct paths in $\mathcal{C}$ are temporally disjoint.
\end{definition}

\begin{definition}[Temporal antichain]\label{def:antichain}
    Let $G$ be a temporal digraph. A \emph{temporal antichain} of~$G$ is a set of vertices in $G$ that are pairwise not temporally connected.
\end{definition}

For a temporal digraph $G$, we denote by $\mathrm{TPC}(G)$ the minimum cardinality of a temporal path cover of $G$, by $\mathrm{TDPC}(G)$ the minimum cardinality of a temporally disjoint path cover of $G$, and by $\mathrm{TAC}(G)$ the maximum cardinality of a temporal antichain of~$G$.
Based on the above definitions, we define the following computational problems for input temporal digraph $G$ and positive integer $k$.
\begin{itemize}
    \item \TPC\ (\TPCshort): decide whether there exists a temporal path cover of size~$k$ in $G$.
    \item \TDPC\ (\TDPCshort): decide whether there exists a temporally disjoint path cover of size~$k$ in $G$.
    \item \TAC\ (\TACshort): decide whether there exists a temporal antichain of size~$k$ in $G$.
\end{itemize}

Observe that \TPCshort\ and \TACshort\ are \emph{dual} in a way reminiscent of \EC\ and \IS. Indeed, every temporal path cover is at least as large as the maximum temporal antichain, as each temporal path can cover at most one element of any temporal antichain. From this observation, combined with the definition of \TDPCshort, we obtain inequality~\eqref{ineq:AC_TPC}, which is central throughout the paper.
\begin{equation}\label{ineq:AC_TPC}
    \mathrm{TAC}(G) \leq \mathrm{TPC}(G) \leq \mathrm{TDPC}(G).
\end{equation}

\begin{definition}[Dilworth property]\label{def:dilworth}
    A class $\mathcal{C}$ of temporal graphs has the Dilworth property (resp.~TD-Dilworth property) if for each graph $G$ in $\mathcal{C}$ it holds that $\mathrm{TAC}(G) = \mathrm{TPC}(G)$ (resp. $\mathrm{TAC}(G) = \mathrm{TDPC}(G)$).
\end{definition}

We refer to \cref{table:dilworth_classes} for an overview of temporal graph classes that are known to satisfy or violate the (TD-)Dilworth property.
We next define the connectivity graph of a temporal graph, which can be seen as the undirected temporal equivalent of the transitive closure graph of a static graph.

\begin{definition}[Connectivity Graph~\cite{temporal_path_cover_Dailly_MFCS}]\label{def:connectivity_graph}
    Let $G=(V,E_1,\ldots,E_{t_{\max}})$ be a temporal digraph. The \emph{connectivity graph} of $G$, denoted by $C_{G}$, is the static undirected graph with vertex set $V(C_{G})=V,$ and edge set $ E(C_{G}) = \{(u,v) \mid u\neq v \text{ and } u, v \text{ are temporally connected in } G\}.$
\end{definition}

\begin{figure}[ht]
     \centering
     \begin{subfigure}[b]{0.48\textwidth}
         \centering
         \includegraphics[width=\linewidth, page=1]{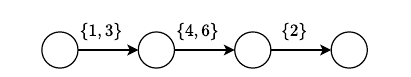}
         \caption{Temporal graph.}
         \label{fig:connectivity_transitivity_1}
     \end{subfigure}\hfill
     \begin{subfigure}[b]{0.48\textwidth}
         \centering
         \raisebox{-0.9cm}{
         \includegraphics[width=\linewidth, page=2]{figures/temporal.drawio.pdf}
         }
         \caption{Connectivity graph.}
         \label{fig:connectivity_transitivity_2}
     \end{subfigure}
     \caption{An example where the connectivity graph of a temporal graph is not transitively closed.}
     \label{fig:connectivity_transitivity}
\end{figure}



The connectivity graph $C_{G}$ can be constructed in $\bigO(\pi n^3)$ time. Indeed, for each source vertex, an adaptation of Dijkstra's algorithm can compute earliest-arrival times to all other vertices in $\bigO(n^2+\pi m)$ time (cf.~\cite{Path_Wu}). 
Repeating for all $n$ sources gives the claimed bound, since $m=\bigO(n^2)$.
Note that, in contrast to static graphs, the connectivity graph of a temporal graph is not necessarily transitively closed (see \cref{fig:connectivity_transitivity} for an example). Intuitively, this is likely one of the reasons for the difference in computational complexity between static and temporal path cover variants: the known polynomial-time algorithm for \PC\ in DAGs~\cite{Fulkerson1956} exploits the aforementioned transitive property, while \TPCshort\ is \NP-hard even on temporal DAGs~\cite{temporal_path_cover_Dailly_MFCS}.




\subsection{Parameterizations}\label{subsec:prelims_parameterized}

Throughout the paper, we assume familiarity with basic notions of \emph{parameterized complexity} (cf.~\cite{books/DowneyF13}).
All parameters used for temporal digraphs in the context of this paper always refer to the underlying \emph{undirected} graph of the input temporal digraph $G$, which we denote by~$\underline{U_G}$. Hence, the classic definitions of parameters for static undirected graphs suffice. Most of our results pertain to the following classical graph parameters.

For a static undirected graph $G$, we denote its \emph{treewidth} with $\tw(G)$; its \emph{pathwidth} with $\pw(G)$; its \emph{feedback vertex set number} (i.e., the minimum number of vertices that have to be deleted in order for the graph to become acyclic) with $\fvs(G)$; its \emph{maximum degree} with~$\Delta(G)$. 



We will also often use the following parameter, which intuitively indicates the graph's similarity to a forest of lines. Although not a standard parameter used in the literature, we argue that it is a natural parameter for the studied problems. Moreover, it enables us to indirectly obtain hardness results for some of the aforementioned standard parameters.

\begin{definition}[Deletion distance to linear forest]\label{def:ddlf}
    For a static undirected graph $G$, we define $\ddlf(G)$ as the minimum number of vertices that have to be deleted from $G$ to obtain a \emph{linear forest}, i.e., a union of disjoint paths.
\end{definition}

We now explain why $\ddlf(G)$ is, in fact, a very natural parameter for path cover problems. For graphs with $\ddlf(G)=0$ (i.e., linear forests), \PC\ is trivial and the cardinality of its solution is equal to the number of connected components. Hence, it is natural to expect that path cover problems may be tractable when $\ddlf(G)$ is small. We remark that, for any graph~$G$, it holds that $\fvs(G) \le \ddlf(G)$ and $\pw(G) \le \ddlf(G)+1$,
since a linear forest is a forest and has pathwidth~$1$. 
Therefore, the hardness results of \cref{thrm:restricted_length,thrm:TPC_tw_tau_hardness,thrm:TPC_pathwidth} imply analogous hardness for the more standard parameters
pathwidth and feedback vertex set number, thus also for treewidth.

We remark that computing $\ddlf(G)$ is NP-hard, as it is the vertex-deletion distance to a hereditary property~\cite{jcss/LewisY80}; nevertheless it admits a straightforward FPT algorithm: one can simply branch on the vertices of degree at least $3$ and (i) either delete the vertex, or (ii) out of three of its neighbors, choose one to delete. The resulting graph is a union of paths and cycles, thus one needs to additionally delete a vertex per cycle. Consequently, the final running time is $4^{\ddlf} n^{O(1)}$.

\section{Introductory results}\label{sec:introductory_results}


\subsection{Hardness of \TAC}\label{sec:antichain}
\TPCshort\ and \TDPCshort\ are known to be \NP-hard on temporal DAGs~\cite{temporal_path_cover_Dailly_MFCS}. We complement this with a simple reduction that proves \NP-hardness for \TACshort\ on temporal DAGs.

\begin{lemmarep}[\appsymb]\label{lem:hardness_max_antichain}
    \TACshort\ is \NP-hard even on temporal DAGs with $t_{\max} = 1$.
\end{lemmarep}

\begin{proof}
We reduce from \IS. 
Let $H=(V,E)$ be a static undirected graph, and fix an arbitrary ordering on its vertices $v_1,\dots,v_n$.
We construct a temporal digraph $G=(U_G,\lambda)$ on the vertex set $V$ as follows: for every edge $\{v_i,v_j\}\in E$ with $i<j$, add the directed edge $(v_i,v_j)$ to $U_G$ and set $\lambda\bigl((v_i,v_j)\bigr)=\{1\}$.

Since every edge of $G$ is directed according to the ordering of the vertices, $U_G$ is a DAG, and thus $G$ is a temporal DAG. 
Since all edges are active only at time $1$, every temporal path in $G$ has length at most $1$. Therefore, two vertices are temporally connected in $G$ if and only if they are adjacent in $H$. It follows that a set $S \subseteq V$ is a temporal antichain in $G$ if and only if $S$ is an independent set in $H$. Thus, an independent set of size $k$ exists in $H$ if and only if a temporal antichain of size $k$ exists in~$G$, for all $k\in \mathbb{N}$.
\end{proof}


\subsection{Hardness of verifying the (TD-)Dilworth property}\label{sec:dilworth_decision_hard}

Here, we briefly show that checking whether the (TD-)Dilworth property holds is \NP-hard. Thus, in the context of our paper, the (TD-)Dilworth property is generally given as a \emph{promise}, which is computationally hard to verify by itself. 

\begin{lemmarep}[\appsymb]\label{lem:Dilworth_verification_hard}
    Verifying whether a temporal graph satisfies the Dilworth property is \NP-hard. The same also holds for the TD-Dilworth property.
\end{lemmarep}

\begin{proof}
    Let $H=(V,E)$ be a static connected undirected graph.
    We construct a temporal graph $G$ with $U_G:=H$, where every edge is bidirected and has label set $[n]$. 
    Since all paths in $H$ have length bounded by $n$, it follows that every path in $H$ corresponds to a valid temporal path in $G$ and vice versa.    
    Since $H$ is connected, it follows that $\mathrm{TAC}(G)=1$. As such, $G$ satisfies both the Dilworth and TD-Dilworth properties if and only if all of its vertices can be covered by one temporal path. We conclude that verifying whether $G$ satisfies the Dilworth (or TD-Dilworth) property is equivalent to solving \textsc{Hamiltonian Path} on $H$, which is a classical \NP-complete problem.
\end{proof}

\begin{remark}
    Regardless of \cref{lem:Dilworth_verification_hard}, it is possible to check whether specific standard temporal graph classes (such as temporal oriented trees) satisfy the (TD-)Dilworth property. See \cref{table:dilworth_classes} for the results of Chakraborty et al.~\cite{temporal_path_cover_Dailly_MFCS}.
\end{remark}

\section{T(D)PC under Dilworth Properties}\label{sec:Dilworth_Lovasz}

In this section we study \TPCshort\ and \TDPCshort\ under the promise of the Dilworth and TD-Dilworth properties respectively. We show a contrast between the decision and search versions of both problems: computing the size of the respective path cover can be done in polynomial time, while recovering the cover itself would imply $\P=\NP$. For the former, we prove that the size is exactly equal to the \emph{Lov{\'a}sz number} of the connectivity graph.

\subsection{A Polynomial-time algorithm for the decision version}\label{subsec:Dilworth_decision}

For a static graph $G$, we denote by $\alpha(G)$, $\omega(G)$, $\chi(G)$, $cc(G)$, and $\vartheta(G)$ its maximum independence number, 
clique number, chromatic number, minimum clique cover size, and Lov{\'a}sz number respectively. 
The Lovász number of a graph, also known as the Lovász theta function, is a real number functioning as an upper bound on the Shannon capacity of the graph. It was introduced by Lovász in 1979~\cite{Lovasz_number_original}. Interestingly, it is sandwiched between hard-to-compute graph parameters, as shown in the following theorem.

\begin{theorem}[Lov{\'a}sz Sandwich Theorem \cite{Knuth_sandwich,Lovasz_sandwich}]\label{thrm:lovasz}
For every graph $G$ it holds that \[\omega(G)\leq \vartheta(\overline G)\leq\chi(G),\] where $\overline G$ denotes the complement of the graph $G$.
\end{theorem}

The next observation follows directly from the definition of the connectivity graph~$C_G$.  

\begin{observation}\label{obs:clique}
    Let $P$ be a temporal path in $G$. The vertices in $V(P)$ form a clique in~$C_G$.
\end{observation}

\begin{lemma}\label{lem:lovasz_inequality_chain}
Let $G$ be a temporal digraph and let $C_G$ be its connectivity graph. Then, \[\mathrm{TAC}(G) = \alpha(C_G) \leq \vartheta(C_G) \leq cc(C_G) \leq \mathrm{TPC}(G) \leq \mathrm{TDPC}(G).\]
\end{lemma}
\begin{proof}
By definition of $C_G$, two vertices are adjacent in $C_G$ if and only if (at least) one temporally reaches the other in $G$. It follows that a set $A\subseteq V$ is a temporal antichain in~$G$ if and only if it is an independent set in $C_G$ and, thus, $\mathrm{TAC}(G) = \alpha(C_G)$. By the definitions of the problems it follows that for a graph $H$ and its complement $\overline H$, $\chi(\overline H)=cc(H)$ and $\omega(\overline H)=\alpha(H)$; in conjunction with \cref{thrm:lovasz} and the fact that $\overline{\overline{H}} = H$, we obtain $\alpha(C_G) \leq \vartheta(C_G) \leq cc(C_G)$. By \cref{obs:clique} it follows that $cc(C_G) \leq \mathrm{TPC}(G)$. Finally, the rightmost inequality is already established in~\eqref{ineq:AC_TPC}.
\end{proof}

\begin{theorem}\label{thrm:dilworth-polytime}
Let $G$ be a temporal digraph. If $G$ satisfies the Dilworth property, i.e., $\mathrm{TAC}(G) = \mathrm{TPC}(G)$, then the common value $\mathrm{TAC}(G) = \mathrm{TPC}(G)$ can be computed in polynomial time.
\end{theorem}

\begin{proof}
First, recall that $C_G$ can be constructed in polynomial time (see \cref{subsec:prelims_problems}).
By \cref{lem:lovasz_inequality_chain} we have \[\mathrm{TAC}(G) = \alpha(C_G) \leq \vartheta(C_G) \leq cc(C_G) \leq \mathrm{TPC}(G).\] 
Since $\mathrm{TAC}(G) = \mathrm{TPC}(G)$, all inequalities collapse. Hence, $\mathrm{TAC}(G)=\vartheta(C_G)=\mathrm{TPC}(G)$. 
From this, it follows that $\vartheta(C_G)$ is an integer. 
The Lov{\'a}sz number $\vartheta(C_G)$ can then be computed with the ellipsoid method, using the algorithm $\operatorname{THETA}$ of~\cite{Lovasz_perfect}, which, given a graph $H$ and an error parameter $\varepsilon > 0$, outputs a rational number $T$ satisfying $|T-\vartheta(H)|<\varepsilon$ in polynomial time.\footnote{The function also receives a weight vector as input, which is a list of ones in our case.}  
Applying this algorithm to $C_G$ with $\varepsilon = 1/4$ and rounding to the closest integer returns the exact value of $\vartheta(C_G)$. 
\end{proof}

The above result also holds for \TDPCshort\ when the TD-Dilworth property is promised instead. The proof is essentially identical to that of \cref{thrm:dilworth-polytime} and is thus omitted.

\begin{theorem}\label{thrm:dilworth-polytime_disjoint}
Let $G$ be a temporal digraph. If $G$ satisfies the TD-Dilworth property, i.e., $\mathrm{TAC}(G)=\mathrm{TDPC}(G)$, then the common value $\mathrm{TAC}(G) = \mathrm{TDPC}(G)$ can be computed in polynomial time.
\end{theorem}

For the sake of completeness, we briefly show that the statement in \cref{thrm:dilworth-polytime} is not an equivalence. Similarly, the same also holds for \cref{thrm:dilworth-polytime_disjoint}.

\begin{propositionrep}[\appsymb]
The converse of \cref{thrm:dilworth-polytime} does not hold, i.e., there exists a class of temporal graphs for which the value $\mathrm{TPC}(G)$ can be computed in polynomial time, while violating the Dilworth property.
\end{propositionrep}

\begin{proof}
Let $C_1$ be the class of temporal digraphs in which every edge is active
only at time~$1$. Observe that every temporal path in a graph
of $C_1$ consists of a single edge. Thus, for $G\in C_1$, $\mathrm{TPC}(G)$ equals the size of the minimum edge cover of $\underline{U_G}$, which can be computed in polynomial time~\cite{GareyJohnson}. However, $C_1$ does not satisfy the Dilworth property. Consider the temporal DAG $D$ with $ V=\{a,b,c\}$ and $E=\{(a,b),(a,c),(b,c)\}$, all active only at time $1$. Observe that $\mathrm{TAC}(D)=1<2=\mathrm{TPC}(D)$. Note that this example can be rescaled to an arbitrary size by adding copies of this triangle, all sharing only vertex $a$.
\end{proof}

\subsection{Hardness of the search version}\label{subsec:Dilworth_search}

\cref{thrm:dilworth-polytime} gives a polynomial-time algorithm for computing the size of a minimum temporal path cover for graphs satisfying the Dilworth property.
A natural question is whether a polynomial-time algorithm can also recover the cover itself.
Interestingly, we show that this is not the case, unless $\P = \NP$.
Our approach is reminiscent of standard polynomial-time reductions, with the particularity that the hypothetical polynomial-time algorithm (the existence of which would show $\P = \NP$) does not need to work correctly
on all the reduced instances. 

\begin{theoremrep}[\appsymb]\label{thrm:TPC_search_dilworth}
    Unless $\P = \NP$, there is no polynomial-time algorithm for the search version of \TPCshort, i.e., computing a minimum-cardinality temporal path cover, even when restricted to temporal DAGs with the Dilworth property.
\end{theoremrep}

\begin{proof}
    Let us assume the existence of an algorithm $A$ that constructs in polynomial time a minimum temporal path cover for temporal DAGs that satisfy the Dilworth property. Thus, when the Dilworth property holds, $A$ runs in time $p(|G|)$, where $p$ is some polynomial and $|G|$ is the total encoding size of the input graph $G$. When the Dilworth property is violated, we do not impose any restriction on the behavior of $A$, i.e., it may output an invalid solution, exceed $p(|G|)$ time, or even not terminate.
    We will use this hypothetical algorithm $A$ to design a polynomial-time algorithm for \TPCshort\ on temporal DAGs, that is, deciding whether an arbitrary temporal DAG~$G$ admits a temporal path cover of size~$k$, where $1\leq k\leq |V(G)|$. \TPCshort\ is known to be \NP-hard on temporal DAGs~\cite{temporal_path_cover_Dailly_MFCS}; thus, we will show that $A$ cannot exist if $\P \neq \NP$.

    Let $G$ be a temporal DAG with time-steps $\{1,\ldots,t_{\max}\}$.   
    Construct a temporal graph~$G_k$ by padding $G$ with a set $S_k$ of $k$ new vertices, and add directed edges from every vertex of~$V(G)$ to every vertex of $S_k$ with a new time-step $t_{\max}+1$. Notice that $G_k$ is also a temporal DAG.
    
    First, observe that the vertices of $S_k$ are sinks (i.e., they have out-degree $0$). Since every temporal path in $G_k$ can have at most one sink, we obtain $\mathrm{TPC}(G_k) \geq k$. By definition of~$G_k$, it also follows that $\mathrm{TPC}(G_k) \geq \mathrm{TPC}(G)$. Next, observe that any temporal path cover of $G$ with $\ell$ paths can be extended to a temporal path cover of $G_k$ by appending a distinct vertex of~$S_k$ to $\min\{\ell,k\}$ of its paths, and covering any remaining vertices of~$S_k$ with singleton paths. The extensions are valid temporal paths because all new edges are active at time $t_{\max}+1$. The proposed cover of $G_k$ contains $\max\{\ell,k\}$ paths. It follows that $\mathrm{TPC}(G_k)\leq\max\{\mathrm{TPC}(G),k\}$. From all the above we obtain
    \begin{equation}\label{eq:TPC_equality}
        \mathrm{TPC}(G_k)=\max\{\mathrm{TPC}(G),k\}.
    \end{equation}    

    Next, observe that the vertices of $S_k$ form a temporal antichain.
    Since every vertex in $S_k$ is temporally reachable from any vertex in $V(G)$, any temporal antichain of $G_k$ is either a subset of $V(G)$ or of $S_k$.
    Thus, we obtain
    \begin{equation}\label{eq:antichain_equality}
        \mathrm{TAC}(G_k)=\max\{\mathrm{TAC}(G),k\}.
    \end{equation}

    We will now use algorithm $A$ to decide whether $\mathrm{TPC}(G)\leq k$. We run algorithm $A$ on~$G_k$ for $p(|G_k|)$ steps. Note that $G_k$ does not necessarily satisfy the Dilworth property, and thus algorithm $A$ may exceed the allocated time. We accept if and only if $A$ outputs a valid temporal path cover of $G_k$ with at most $k$ paths within $p(|G_k|)$ steps. Note that the validity of the cover can trivially be checked in polynomial time.

    We now prove the correctness of the described procedure, i.e., that it accepts if and only if $\mathrm{TPC}(G)\leq k$.
    \begin{itemize}
        \item First, consider the case $\mathrm{TPC}(G)\leq k$. By equations~\eqref{ineq:AC_TPC},~\eqref{eq:TPC_equality},~\eqref{eq:antichain_equality}, we obtain that $\mathrm{TPC}(G_k)=\mathrm{TAC}(G_k)=k$, meaning that $G_k$ satisfies the Dilworth property. Thus, by assumption, algorithm $A$ must terminate within $p(|G_k|)$ steps and correctly output a valid temporal path cover of $G_k$ with exactly $k$ paths. As such, our procedure will accept.
        \item Second, consider the case $\mathrm{TPC}(G) > k$. By definition of $G_k$ we have $\mathrm{TPC}(G_k) \geq \mathrm{TPC}(G) > k$, meaning that $A$ cannot output a valid temporal path cover of $G_k$ with at most $k$ paths.
        Consequently, our procedure will reject.
    \end{itemize}

    Since the construction runs in polynomial time, we obtain a polynomial-time algorithm for an $\NP$-hard problem~\cite{temporal_path_cover_Dailly_MFCS}, which is a contradiction under $\P \neq \NP$. Thus, unless $\P = \NP$, there can be no polynomial-time algorithm for \TPCshort, even when restricted to temporal DAGs with the Dilworth property. 
\end{proof}

An analogous result also holds for \TDPCshort\ on graphs with the TD-Dilworth property. The proof is almost identical to that of \cref{thrm:TPC_search_dilworth}, and is thus omitted.

\begin{theorem}\label{thrm:TDPC_search_dilworth}
    Unless $\P = \NP$, there is no polynomial-time algorithm for the search version of \TDPCshort, i.e., computing a minimum-cardinality temporally disjoint path cover, even when restricted to temporal DAGs with the TD-Dilworth property.
\end{theorem}

\section{T(D)PC under relaxed Dilworth Properties}\label{sec:tpc-two-paths}

A natural question arising from \cref{thrm:dilworth-polytime,thrm:dilworth-polytime_disjoint} is whether the values $\mathrm{TPC}(G)$ and $\mathrm{TDPC}(G)$ can also be computed efficiently under a relaxed version of the (TD-)Dilworth promise, i.e., a promise that the values $\mathrm{TAC}(G)$ and $\mathrm{TPC}(G)$ (or $\mathrm{TDPC}(G)$) have small difference. In this section we establish that this question is, in fact, \NP-hard for both problems even with the promises $\mathrm{TPC}(G)-\mathrm{TAC}(G) \leq 1$ and $\mathrm{TDPC}(G)-\mathrm{TAC}(G) \leq 1$ respectively, thus showing a sharp contrast with \cref{thrm:dilworth-polytime,thrm:dilworth-polytime_disjoint}.

We start by proving an auxiliary lemma for static graphs, which essentially states that, even when it is known that a static undirected graph can be covered by two vertex-disjoint paths, it is \NP-complete to decide whether it can be covered by just one (i.e., a Hamiltonian path). To this end we utilize a reduction from \textsc{Hamiltonian Cycle} to \textsc{Hamiltonian Path}, together with the fact that the former is known to be \NP-complete even when a Hamiltonian path is given in the input~\cite[GT37]{GareyJohnson} (see also~\cite{stoc/PapadimitriouS76}). See \cref{fig:path-cover-reduction} for a sketch of the reduction.

\begin{figure}[ht]
    \centering
    \begin{tikzpicture}[
        font=\small,
        vertex/.style={circle, draw=black!65, semithick, fill=white,
            inner sep=0pt, minimum size=6.5mm},
        first/.style={vertex, draw=teal!75!black, fill=teal!10},
        second/.style={vertex, draw=orange!80!black, fill=orange!15},
        path/.style={draw=teal!75!black, line width=1.15pt},
        short path/.style={draw=orange!80!black, line width=1.15pt,
            dash pattern=on 3pt off 2pt},
        extra/.style={draw=black!45, semithick},
        remainder/.style={draw=black!20, fill=black!3,
            rounded corners=4pt},
        path label/.style={font=\footnotesize, inner sep=2pt},
        heading/.style={anchor=west, font=\small}
    ]
        \node[heading] at (-0.35,1.9) {\textbf{(a)} Input $F$};
        \draw[remainder] (0.75,-0.5) rectangle (4.65,0.55);
        \node[font=\footnotesize, text=black!65, anchor=north east]
            at (4.6,-0.58) {$F-v_1$};
        \node[first] (v1) at (0,0) {$v_1$};
        \node[first] (v2) at (1.2,0) {$v_2$};
        \node[first] (vj) at (2.7,0) {$v_j$};
        \node[first] (vn) at (4.2,0) {$v_n$};
        \node (dots1) at (1.95,0) {$\cdots$};
        \node (dots2) at (3.45,0) {$\cdots$};
        \draw[path] (v1) -- node[path label, below] {$P$} (v2);
        \draw[path] (v2) -- (dots1);
        \draw[path] (dots1) -- (vj);
        \draw[path] (vj) -- (dots2);
        \draw[path] (dots2) -- (vn);
        \draw[extra] (v1.north) .. controls (0,1.4) and (2.7,1.4)
            .. (vj.north);

        \draw[->, >=stealth, thick, draw=black!65] (4.95,0) -- (5.85,0);
        \node[font=\scriptsize, align=center, text=black!65]
            at (5.4,0.65) {split $v_1$\\add $s,t$};

        \begin{scope}[xshift=6.5cm]
            \node[heading] at (-0.35,1.9) {\textbf{(b)} Output $H$};
            \draw[remainder] (1.95,-0.5) rectangle (5.85,0.55);
            \node[font=\footnotesize, text=black!65, anchor=north east]
                at (5.8,-0.58) {$F-v_1$};
            \node[first] (s) at (0,0) {$s$};
            \node[first] (a) at (1.1,0) {$a$};
            \node[second] (t) at (0,-1.35) {$t$};
            \node[second] (b) at (1.1,-1.35) {$b$};
            \node[first] (w2) at (2.4,0) {$v_2$};
            \node[first] (wj) at (3.9,0) {$v_j$};
            \node[first] (wn) at (5.4,0) {$v_n$};
            \node (dots3) at (3.15,0) {$\cdots$};
            \node (dots4) at (4.65,0) {$\cdots$};
            \draw[path] (s) -- node[path label, below] {$P_1$} (a);
            \draw[path] (a) -- (w2);
            \draw[path] (w2) -- (dots3);
            \draw[path] (dots3) -- (wj);
            \draw[path] (wj) -- (dots4);
            \draw[path] (dots4) -- (wn);
            \draw[short path] (t) -- node[path label, below] {$P_2$} (b);
            \draw[extra] (a.north) .. controls (1.1,1.4) and (3.9,1.4)
                .. (wj.north);
            \draw[extra] (b.east) to[out=0,in=-90] (w2.south);
            \draw[extra] (b.south east) .. controls (2.4,-1.85) and (3.9,-1.5)
                .. (wj.south);
        \end{scope}
    \end{tikzpicture}
    \caption{The reduction used in \cref{lem:static-two-paths}. The input graph $F$ admits a Hamiltonian path $(v_1,\ldots,v_n)$.
        Vertex $v_1$ is replaced by vertices $a$ and $b$, each connected with one leaf, as well as all of the neighbors of $v_1$ in $F$.
        The shaded subgraph $F-v_1$ is unchanged. The solid teal path
        $P_1=(s,a,v_2,\ldots,v_n)$ and dashed orange path $P_2=(t,b)$
        form a vertex-disjoint path cover of $H$.}
    \label{fig:path-cover-reduction}
\end{figure}

\begin{lemmarep}[\appsymb]\label{lem:static-two-paths}
    Given a static connected undirected graph $H$ together with a path cover consisting of two vertex-disjoint paths, deciding whether $H$ admits a Hamiltonian path is \NP-complete. 
\end{lemmarep}

\begin{proof}
    We reduce from \textsc{Hamiltonian Cycle}, which remains \NP-complete
    when a Hamiltonian path is given in the input~\cite[GT37]{GareyJohnson} (see also~\cite{stoc/PapadimitriouS76}).
    Let $F$ be a static undirected connected graph with Hamiltonian path $(v_1,\ldots,v_n)$,
    where $n\geq 3$.
    Replace $v_1$ by two vertices $a,b$, each adjacent to every
    neighbor of $v_1$, and add leaves $s,t$ adjacent only to $a$ and $b$ respectively. Retain all edges between vertices other than $v_1$ and call the resulting graph $H$ (see \cref{fig:path-cover-reduction}). Observe that $H$ is connected and the paths
    \[
        P_1=(s,a,v_2,\ldots,v_n)
        \qquad\text{and}\qquad
        P_2=(t,b)
    \]
    form a vertex-disjoint cover of $H$.

    We now prove the correctness of the reduction, i.e., that there exists a Hamiltonian cycle in $F$ if and only if there exists a Hamiltonian path in $H$.

    ($\Rightarrow$) Suppose $F$ admits a Hamiltonian cycle $C_F$. By definition, $C_F$ contains exactly two edges incident to $v_1$; let us name these edges $(v_1,v_i)$ and $(v_1,v_j)$. Design a path in $H$ by taking $C_F$ and removing these two edges; replace the former with edges $(s,a)$ and $(a,v_i)$, and replace the latter with edges $(t,b)$ and $(b,v_j)$. The suggested edges exist due to construction. It is clear that the designed path of $H$ is a Hamiltonian path.

    ($\Leftarrow$) Suppose $H$ admits a Hamiltonian path $P_H$. Since $s$ and $t$ are leaves, they have to be the two endpoints of $P_H$. By definition, $P_H$ contains exactly two edges incident to $a$, one of them being $(s,a)$, and exactly two edges incident to $b$, one of them being $(t,b)$. Let $(a,v_i)$ and $(b,v_j)$ be the rest of the aforementioned edges. Design a cycle in $F$ by removing all four aforementioned edges and replacing them with $(v_1,v_i)$ and $(v_1,v_j)$. The suggested edges exist due to construction. It is clear that the designed cycle of $F$ is a Hamiltonian cycle.
    
    The graph $H$ and its cover $\{P_1,P_2\}$ can be constructed
    in $O(|V(F)|+|E(F)|)$ time, hence the suggested reduction is polynomial-time.
\end{proof}

The reduction of \cref{lem:static-two-paths} can be adapted to temporal graphs by simply placing the label set $[n+3]$ to all edges of the output graph and making them bidirected. Since all paths in the output graph have length bounded by $n+3$ (cf. \cref{fig:path-cover-reduction}), exactly the same arguments used in the proof of \cref{lem:static-two-paths} hold. The proof of the following lemma is thus omitted. We remark that, since the source static graph is connected due to admitting a Hamiltonian path by assumption, the constructed temporal graph has maximum temporal antichain size $1$.

\begin{lemma}\label{lem:temporal-two-paths}
    Given a temporal graph $G$ with $\mathrm{TAC}(G)=1$, together with a temporal path cover consisting of two vertex-disjoint\footnote{As shown in \cref{fig:path-cover-reduction}, the suggested paths are \emph{vertex-disjoint}, which is stronger than \emph{temporally disjoint}.} temporal paths, deciding whether $G$ can be covered by a single temporal path is \NP-complete. 
\end{lemma}
%
From \cref{lem:temporal-two-paths} we obtain the main result of this section, contrasting \cref{thrm:dilworth-polytime,thrm:dilworth-polytime_disjoint} and yielding an interesting dichotomy.
\begin{theorem}\label{thm:difference_hardness}
    Given a temporal graph $G$, computing the value $\mathrm{TPC}(G)$ is \NP-hard even under the promise that $\mathrm{TPC}(G)-\mathrm{TAC}(G) \leq 1$. The same holds for $\mathrm{TDPC}(G)$ under the promise that $\mathrm{TDPC}(G)-\mathrm{TAC}(G) \leq 1$.
\end{theorem}

\section{Parameterized complexity of \TPCshort}\label{sec:TPC_parameterized}

\subsection{W[1]-Hardness parameterized by deletion distance to linear forest with constant number of time-steps}

As already mentioned, \TPCshort\ admits an \XP\ algorithm parameterized by
$\tw(\underline{U_G})+t_{\max}$~\cite{temporal_path_cover_Dailly_MFCS}.
We show that an \FPT\ algorithm for this parameterization is unlikely under standard complexity assumptions, settling an open question from~\cite{temporal_path_cover_Dailly_MFCS}.
In fact, more strongly, we prove that the problem with $t_{\max}=2$ already remains \W$[1]$-hard when parameterized by $\ddlf(\underline{U_G})$ (we refer to Section~\ref{subsec:prelims_parameterized} for an explanation of why this result is stronger).

We reduce from \RUBP, where we seek to partition items into bins in a way that each
bin sums up to the same value $B$, while each item chooses between exactly two bins.
Since each item is eligible for exactly two bins, each item gadget
has two possible minimum-cover configurations, one for each choice.
A single selector vertex~$\sigma$, fed by $n$ leaves, forces exactly one of these configurations
for each of the $n$ items. Choosing bin~$c$ for an item of size $s$ leaves
exactly $s$ assignment vertices associated with $c$ to be covered by the bin
gadgets. Each bin gadget provides $B$ paths, each of which can cover one such
vertex. Hence a cover of the target size exists precisely when every bin
receives items of total size $B$. Lastly, deleting the selector and the bin
hubs leaves only disjoint paths.

\begin{theorem}\label{thrm:TPC_tw_tau_hardness}
    \TPCshort\ is \W$[1]$-hard parameterized by $\ddlf(\underline{U_G})$, even on temporal DAGs with $t_{\max}=2$ in which every
    edge is active at exactly one time-step.
\end{theorem}

\begin{proof}
    We reduce from \RUBP. An instance consists of $n$ items with unary sizes
    $S=(s_1,\ldots,s_n)$ and $k$ bins. Each item $i\in[n]$ has a set
    $f(i)\in\binom{[k]}{2}$ of two eligible bins.
    The question is whether every item can be assigned to an eligible bin so
    that each bin receives total size exactly $B \coloneq \sum_{i\in[n]}s_i / k$.
    The problem is
    \W$[1]$-hard parameterized by $k$~\cite{mfcs/HanakaLVY24}. We may assume
    that $B$ is an integer, since otherwise the instance is negative.
    We now construct a temporal digraph $G$ and an integer $K$ such that the \RUBP\
    instance is positive if and only if $G$ has a temporal path cover of size at most~$K$.

    \proofsubparagraph{Construction.}
    Fix an item $i\in[n]$, write $f(i)=\{c,d\}$ with $c<d$, and let $s \coloneq s_i$.
    Introduce sources $u_i^1,\ldots,u_i^s$, sinks $r_i^0,\ldots,r_i^s$,
    and two sets $X_i^c \coloneq \setdef{x_i^{c,j}}{j\in[s]}$ and $X_i^d \coloneq \setdef{x_i^{d,j}}{j\in[s]}$
    of assignment vertices associated with bins $c$ and $d$, respectively.
    For every $j\in[s]$, add $u_i^j\xrightarrow{1}x_i^{d,j}\xrightarrow{2}r_i^{j-1}$
    and $u_i^j\xrightarrow{1}x_i^{c,j}\xrightarrow{2}r_i^j$,
    where each arrow is labelled by its unique active time-step.

    Next, introduce one selector vertex $\sigma$ and $n$ selector sources
    $\ell_1,\ldots,\ell_n$. For every $i \in [n]$, add
    $\ell_i\xrightarrow{1}\sigma$, $\sigma\xrightarrow{2}r_i^0$, and $\sigma\xrightarrow{2}r_i^{s_i}$.
    Thus a path through $\sigma$ can cover one extreme sink of one item gadget.

    Finally, for every bin $c\in[k]$, introduce a hub $h_c$ and $B$ bin
    sources $z_c^1,\ldots,z_c^B$. Add
    $z_c^q\xrightarrow{1}h_c$ for every $q\in[B]$ and
    $h_c\xrightarrow{2}x$ for every assignment vertex $x$ associated with
    bin~$c$. A path starting at a bin source can therefore cover at most one
    assignment vertex, necessarily one associated with its bin.

    This completes the construction of the temporal digraph $G$;
    part of it is illustrated in \cref{fig:item_gadget}.
    Setting $K \coloneq 2kB+n$ completes the construction of the instance of \TPCshort.

\begin{figure}[ht]
    \centering
    \begin{tikzpicture}[
        >=stealth,
        xscale=1.25, yscale=1.2,
        src/.style={circle, draw, semithick, fill=white,
                    inner sep=0pt, minimum size=6.5mm, font=\small},
        snk/.style={circle, draw, semithick, fill=black!12,
                    inner sep=0pt, minimum size=6.5mm, font=\small},
        binc/.style={circle, draw, semithick, fill=red!25,
                    inner sep=0pt, minimum size=6.5mm, font=\small},
        bind/.style={circle, draw, semithick, fill=blue!18,
                    inner sep=0pt, minimum size=6.5mm, font=\small},
        hub/.style={circle, draw, thick, inner sep=0pt,
                    minimum size=7mm, font=\small},
        core/.style={->, semithick},
        selector/.style={->, semithick, dashed},
        hubedge/.style={->, semithick, dotted},
        tlab/.style={font=\scriptsize, inner sep=1pt, fill=white},
    ]
        \node[snk] (r0) at (0,0) {$r^0$};
        \node[src] (u1) at (1,0) {$u^1$};
        \node[snk] (r1) at (2,0) {$r^1$};
        \node[src] (u2) at (3,0) {$u^2$};
        \node[snk] (r2) at (4,0) {$r^2$};
        \node[src] (u3) at (5,0) {$u^3$};
        \node[snk] (r3) at (6,0) {$r^3$};

        \node[binc] (xc1) at (1.5, 1.15) {$x^{c,1}$};
        \node[binc] (xc2) at (3.5, 1.15) {$x^{c,2}$};
        \node[binc] (xc3) at (5.5, 1.15) {$x^{c,3}$};
        \node[bind] (xd1) at (0.5,-1.15) {$x^{d,1}$};
        \node[bind] (xd2) at (2.5,-1.15) {$x^{d,2}$};
        \node[bind] (xd3) at (4.5,-1.15) {$x^{d,3}$};

        \node[hub] (sigma) at (-1.5,0) {$\sigma$};
        \node[src] (l1) at (-2.5, 0.95) {$\ell_1$};
        \node[src] (l2) at (-2.5, 0.00) {$\ell_2$};
        \node[src] (ln) at (-2.5,-0.95) {$\ell_n$};
        \node at (-2.5,-0.5) {$\vdots$};

        \node[hub, fill=red!25]  (hc) at (3, 2.1) {$h_c$};
        \node[hub, fill=blue!18] (hd) at (3,-2.1) {$h_d$};

        \foreach \j/\rl/\rr in {1/r0/r1, 2/r1/r2, 3/r2/r3}{
            \draw[core] (u\j) -- node[tlab,pos=.55]{$1$} (xc\j);
            \draw[core] (u\j) -- node[tlab,pos=.55]{$1$} (xd\j);
            \draw[core] (xc\j) -- node[tlab,pos=.45]{$2$} (\rr);
            \draw[core] (xd\j) -- node[tlab,pos=.45]{$2$} (\rl);
        }

        \draw[core] (l1) -- node[tlab,pos=.5]{$1$} (sigma);
        \draw[core] (l2) -- node[tlab,pos=.5]{$1$} (sigma);
        \draw[core] (ln) -- node[tlab,pos=.5]{$1$} (sigma);

        \draw[selector] (sigma) to[bend left=10]
            node[tlab,pos=.55]{$2$} (r0);
        \draw[selector] (sigma) to[bend left=95]
            node[tlab,pos=.1]{$2$} (r3);

        \foreach \j in {1,2,3}{
            \draw[hubedge] (hc) to[bend right=8]
                node[tlab,pos=.5]{$2$} (xc\j);
            \draw[hubedge] (hd) to[bend left=8]
                node[tlab,pos=.5]{$2$} (xd\j);
        }

        \begin{scope}[on background layer]
            \node[draw=black!45, dashed, rounded corners, inner sep=9pt,
                  fit=(r0)(r3)(xc1)(xc3)(xd1)(xd3)] (frame) {};
        \end{scope}
        \node[font=\scriptsize, anchor=north west]
            at (frame.south west) {item gadget $i$};
        \node[font=\scriptsize, anchor=north]
            at ([yshift=-2pt]ln.south) {selector sources};
    \end{tikzpicture}
    \caption{The gadget for an item $i$ of size $s_i=3$ eligible for bins
        $\{c,d\}$, with the item index suppressed. Solid edges belong to the
        item gadget, dashed edges come from the shared selector $\sigma$, and
        dotted edges come from the bin hubs. Selecting $r^3$ leaves
        $x^{c,1},x^{c,2},x^{c,3}$ for bin~$c$, whereas selecting $r^0$ leaves $x^{d,1},x^{d,2},x^{d,3}$ for
        bin~$d$. Every edge is labelled by its unique active time-step.}
    \label{fig:item_gadget}
\end{figure}

    \begin{lemmarep}[\appsymb]\label{lem:TPC_packing_to_cover}
        If the \RUBP\ instance is positive,
        then $G$ has a temporal path cover of size $K$.
    \end{lemmarep}

    \begin{proof}
        Consider a feasible assignment of the items.
        For every item $i$, write $f(i)=\{c,d\}$ with $c<d$.

        If item $i \in [n]$ is assigned to bin $c \in [k]$, cover $r_i^{s_i}$ by the path
        $\ell_i\xrightarrow{1}\sigma\xrightarrow{2}r_i^{s_i}$, and for every $j\in[s_i]$, use
        $u_i^j\xrightarrow{1}x_i^{d,j}\xrightarrow{2}r_i^{j-1}$.
        These paths cover all of $X_i^d$ and leave precisely $X_i^c$ uncovered.
        If $i$ is assigned to $d$, cover $r_i^0$ by the path $\ell_i\xrightarrow{1}\sigma\xrightarrow{2}r_i^0$,
        and use $u_i^j\xrightarrow{1}x_i^{c,j}\xrightarrow{2}r_i^j$
        for every $j\in[s_i]$, leaving precisely $X_i^d$ uncovered.

        Consequently, bin~$c$ is associated with exactly as many uncovered
        assignment vertices as the total size assigned to it, namely $B$.
        Match these vertices bijectively with the $B$ sources of bin~$c$ and
        cover them by paths
        $z_c^q\xrightarrow{1}h_c\xrightarrow{2}x$.
        This covers every vertex using $\sum_{i\in[n]}s_i+n+kB=K$ paths, as required.
    \end{proof}

    \begin{lemmarep}[\appsymb]\label{lem:TPC_cover_to_packing}
        If $G$ has a temporal path cover of size at most $K$, then the
        \RUBP\ instance is positive.
    \end{lemmarep}

    \begin{proof}
        Let $\mathcal{C}$ be a temporal path cover of size at most $K$.

        \begin{claim}\label{clm:TPC_tight_cover}
            The cover $\mathcal{C}$ has size exactly $K$, and each of its paths
            contains exactly one source. In particular, every source belongs
            to exactly one path.
        \end{claim}

        \begin{claimproof}
            Every source has in-degree zero, and hence no directed path
            contains two sources. Since $G$ has exactly $K$ sources, every
            cover uses at least $K$ paths. Since $\mathcal{C}$ has at most $K$
            paths and must cover every source, all inequalities are tight.
        \end{claimproof}

        \begin{claim}\label{clm:TPC_selector_choice}
            For every item $i$, exactly one path through $\sigma$ covers a sink
            of its gadget, and that sink is either $r_i^0$ or $r_i^{s_i}$.
        \end{claim}

        \begin{claimproof}
            By \cref{clm:TPC_tight_cover}, every path in $\mathcal{C}$ contains
            a source.
            A bin-source path reaches an assignment vertex at time-step~$2$
            and cannot continue to a sink at the same time-step. Therefore,
            the sinks of gadget~$i$ can be covered only by its $s_i$ item-source
            paths or by paths through $\sigma$. Each item-source path covers at
            most one sink, whereas the gadget contains $s_i+1$ sinks. Hence at
            least one sink of every item gadget is covered through $\sigma$.

            By \cref{clm:TPC_tight_cover}, there are only $n$ paths containing
            selector sources. Since there are $n$ item gadgets, exactly one
            such path enters each gadget. The only sinks adjacent to $\sigma$
            are its two extreme sinks, proving the claim.
        \end{claimproof}

        \begin{claim}\label{clm:TPC_forced_item}
            Let $f(i)=\{c,d\}$ with $c<d$. If the selector path covers
            $r_i^{s_i}$, then the item-source paths leave precisely $X_i^c$
            uncovered. If it covers $r_i^0$, they leave precisely $X_i^d$
            uncovered.
        \end{claim}

        \begin{claimproof}
            Write $s \coloneq s_i$. Suppose first that the selector covers $r_i^s$.
            By \cref{clm:TPC_selector_choice}, the item-source paths must cover
            all of $r_i^0,\ldots,r_i^{s-1}$. The only such path that can cover
            $r_i^0$ starts at $u_i^1$ and passes through $x_i^{d,1}$.
            Inductively, once the paths from
            $u_i^1,\ldots,u_i^{j-1}$ have been fixed, the only remaining
            item-source path that can cover $r_i^{j-1}$ starts at $u_i^j$ and
            passes through $x_i^{d,j}$. Thus all vertices of $X_i^d$ are
            covered and all vertices of $X_i^c$ are left uncovered.

            If the selector covers $r_i^0$, the symmetric backward argument
            starts with $r_i^s$, which forces the path through
            $u_i^s,x_i^{c,s}$, and then forces all item-source paths through
            $X_i^c$. Hence precisely $X_i^d$ is left uncovered.
        \end{claimproof}

        We now assign item $i$ to $c$ if its selector path covers
        $r_i^{s_i}$, and to $d$ if it covers $r_i^0$. By
        \cref{clm:TPC_forced_item}, an item of size $s_i$ leaves exactly $s_i$
        uncovered assignment vertices associated with its chosen bin. These
        vertices can be covered only by bin-source paths. Bin~$c$ has $B$ such
        paths, each covering at most one assignment vertex associated with
        $c$, so the total size assigned to any bin is at most $B$. The total
        size of all items is $kB$; therefore every bin receives total size
        exactly $B$. The resulting assignment is feasible.
    \end{proof}
    Correctness follows by \cref{lem:TPC_packing_to_cover,lem:TPC_cover_to_packing}.
    Furthermore, as all item sizes are given in unary, the reduction runs in polynomial time.

    Every edge is active at exactly one of the two time-steps. Moreover, $G$ is
    a DAG: order first all sources, then $\sigma$ and the bin hubs, then the
    assignment vertices, and finally the sinks. Every edge points forward in
    this ordering. Finally, deleting $\sigma,h_1,\ldots,h_k$ leaves the item
    gadgets as disjoint paths and all selector and bin sources isolated.
    Hence the vertex-deletion distance of $\underline{U_G}$ to a linear forest is at most
    $k+1$. This completes the parameterized reduction.
\end{proof}

\begin{corollaryrep}[\appsymb]\label{cor:TPC_undirected_hardness}
    \TPCshort\ is \W$[1]$-hard parameterized by $\ddlf(\underline{U_G})$, even when $G$ is bidirected, $t_{\max}=2$, and
    every edge is active at exactly one time-step.
\end{corollaryrep}

\begin{proof}
    Regard every edge of the construction as bidirected, without changing its
    time-step. Let $M$ be the set of all item, selector, and bin sources, so
    $|M|=K$. Every vertex of $M$ is incident only with time-step-$1$ edges, no
    edge has both endpoints in $M$, and a temporal path uses at most one edge
    at time-step~$1$. Thus every temporal path contains at most one vertex of
    $M$, and every cover has at least $K$ paths.

    The cover constructed in \cref{lem:TPC_packing_to_cover} remains valid.
    Conversely, a cover of size $K$ contains exactly one vertex of $M$ per
    path. Any path covering a sink must therefore be an item path
    $u_i^j-x-r$ or a selector path $\ell_t-\sigma-r$; a bin path reaches an
    assignment vertex only at time-step~$2$ and cannot continue. Hence the
    arguments of
    \cref{clm:TPC_selector_choice,clm:TPC_forced_item} apply unchanged, and
    yield a feasible packing. The structural bound is unaffected.
\end{proof}

\begin{remark}
    The $t_{\max}=2$ value in \cref{thrm:TPC_tw_tau_hardness,cor:TPC_undirected_hardness} is tight, since \TPCshort\ trivially reduces to \EC\ when $t_{\max}=1$ and can thus be solved in polynomial time.
\end{remark}

\subsection{An FPT algorithm parameterized by vertex cover number plus number of time-steps}

With \cref{thrm:TPC_tw_tau_hardness}, we have proven that \TPCshort\ cannot be \FPT\ by $\tw(\underline{U_G})+t_{\max}$ under standard assumptions. 
We now consider the more restrictive parameterization by vertex cover number instead of treewidth,
and obtain such an \FPT\ algorithm.

Roughly, our algorithm works as follows. First, we identify an \FPT-bounded number of equivalence classes
on the vertices of the independent set, such that any vertex is \emph{interchangeable} by any other
vertex in the same class, for any path of the solution. Consequently, it suffices to keep track of
the sequence of equivalence classes of the vertices that each path traverses, and aim for a solution
that, in total, traverses each equivalence class at least as many times as its cardinality. To do so,
we encode this as an ILP of bounded dimension, which is known to be in \FPT.

\begin{theoremrep}[\appsymb]\label{thrm:TPC_vc_tau}
    Let $G=(V,E_1,\ldots,E_{t_{\max}})$ be a temporal digraph.
    Then \TPCshort\ is \FPT\ parameterized by $\vc(\underline{U_G})+t_{\max}$. This tractability also holds for the search version of \TPCshort, i.e., computing a minimum-cardinality temporal path cover.
\end{theoremrep}

\begin{proof}
    We first compute a minimum vertex cover $X$ of $\underline{U_G}$, and let $I=V\setminus X$.
    In particular, there is no edge between two vertices of $I$ at any time-step.
    Let $\vc \coloneq |X|$.
    We remark that such a minimum vertex cover can be found in \FPT\ time parameterized by $\vc$~\cite[Section 3.1]{books/CyganFKLMPPS15}.

    For $u\in I$, its \emph{temporal signature} records, for every $x\in X$ and
    $t\in[t_{\max}]$, whether $(u,x)\in E_t$ and whether $(x,u)\in E_t$. We partition
    $I$ into classes $\mathcal{Q}=\{Q_1,\ldots,Q_q\}$ of vertices with the same
    temporal signature, where $q\leq 2^{2\vc \cdot t_{\max}}$.

    Note that vertices in the same class are interchangeable in any temporal path: replacing one by another
    preserves all temporal edges between that vertex and $X$. Furthermore, it is easy
    to see that every temporal path contains at most $L \coloneq \min\{2\vc+1,t_{\max}+1\}$ vertices:
    since $I$ is an independent set, between any two vertices of $I$ on the path there
    is a vertex of $X$, and since the edge time-steps are strictly increasing elements
    of $[t_{\max}]$, it has at most $t_{\max}$ edges.

    We define the \emph{type} of a temporal path as its vertex sequence after replacing
    every vertex of $I$ by its class in $\mathcal{Q}$; note that the time-steps are not part of the type.
    We say that a type is \emph{realizable} if its sequence can be instantiated as a temporal path.
    In particular, no vertex of $X$ may occur twice and, if a class $Q_i$ occurs $r$ times, then $|Q_i|\geq r$.
    The latter condition allows us to assign distinct representatives of $Q_i$ to these
    occurrences. Since all vertices of a class have the same temporal
    signature, realizability is independent of the chosen representatives.

    \begin{claim}
        The set $\mathcal{P}$ of realizable path types has size bounded by a
        function of $\vc+t_{\max}$, and it can be enumerated in \FPT\ time.
    \end{claim}

    \begin{claimproof}
        A type has at most $L$ positions, each containing an element of
        $X\cup\mathcal{Q}$. Therefore,
        \[
            |\mathcal{P}|
            \leq \sum_{\ell=1}^{L}(\vc+q)^\ell,
        \]
        which is bounded by a function of $\vc+t_{\max}$. We enumerate these
        sequences and discard those violating the simplicity conditions above.

        It remains to test temporal realizability. A singleton sequence is
        always realizable. For a sequence $(z_1,\ldots,z_\ell)$ with
        $\ell\geq 2$, set $t_0=0$ and, for $j\in[\ell-1]$, let $t_j$ be the
        earliest time-step larger than $t_{j-1}$ at which the directed edge
        $(z_j,z_{j+1})$ is active. Here, an endpoint in $\mathcal{Q}$ may be
        replaced by any representative of that class. If some $t_j$ does not
        exist, then the type is not realizable; otherwise,
        $t_1<\cdots<t_{\ell-1}$ realizes it. This greedy test is correct because
        choosing an edge later can only reduce the available choices for the
        remaining suffix. We store this realizing sequence with the type.
    \end{claimproof}

    Let $\mathcal{R} \coloneq \mathcal{Q}\cup\setdef{\{x\}}{x\in X}$, that is,
    $\mathcal{R}$ is a partition of $V$ consisting of the signature classes
    in $I$ and one singleton class for every vertex of $X$. For every
    $P\in\mathcal{P}$, introduce a variable
    $y_P\in\mathbb{Z}_{\geq 0}$ denoting how many paths of type $P$ are
    selected. For $R\in\mathcal{R}$, let $a_{R,P}$ be the number of
    occurrences of vertices from $R$ in $P$. Consider the integer linear
    program
    \[
        \begin{array}{lll}
        \text{minimize}   & \displaystyle\sum_{P\in\mathcal{P}}y_P  & \\[1ex]
        \text{subject to} &
        \displaystyle\sum_{P\in\mathcal{P}}a_{R,P}y_P\geq |R|       & \text{for every }R\in\mathcal{R},\\[2ex]
                        & y_P\in\mathbb{Z}_{\geq 0}                 & \text{for every }P\in\mathcal{P}.
        \end{array}
        \tag{$\star$}
    \]

    \begin{claim}
        The optimum value of $(\star)$ is equal to the minimum size of a
        temporal path cover of $G$.
    \end{claim}

    \begin{claimproof}
        Given a temporal path cover, group its paths by type and set the
        variables accordingly. Since every vertex in every
        $R\in\mathcal{R}$ is covered, all constraints are satisfied. This gives
        a feasible ILP solution of the same value.

        Conversely, let $y$ be a feasible integer solution and take $y_P$
        copies of every type $P$. Fix a class $R\in\mathcal{R}$. By its
        constraint, there are at least $|R|$ occurrences of $R$ over all
        copies. Order these occurrences copy by copy and assign the vertices
        of $R$ cyclically. Every vertex of $R$ is then used, and occurrences
        within one copy receive distinct vertices because a realizable type
        contains at most $|R|$ occurrences of $R$. Applying this independently
        to every class and using the stored time-steps turns all copies into
        temporal paths covering $V$.

        The resulting cover has size at most
        $\sum_{P\in\mathcal{P}}y_P$; duplicate paths, if any, can be discarded.
        Together with the first direction, this proves that the two optimum
        values are equal.
    \end{claimproof}

    The ILP has $|\mathcal{P}| \le f(\vc,t_{\max})$ variables, for some function $f$.
    It can therefore be solved to optimality in \FPT\ time as \textsc{Integer Linear Programming}
    is \FPT\ parameterized by the number of variables~\cite{Lenstra_ILP}.
    Furthermore, the algorithm returns a minimum temporal path cover as well.
    This completes the proof.
\end{proof}

\section{\NP-hardness of T(D)PC with constant structural parameters}\label{sec:para-NP-hardness}

We complement the results of \cref{sec:TPC_parameterized} with para-\NP-hardness results for \TPCshort\ and \TDPCshort\ for various structural parameters.
As an intermediate step for some of our results, we will prove that \PC\ (i.e., the static graph analogue of \TPCshort, cf.~\cite{Foucaud_static_caldam}) becomes \NP-hard in certain graph classes if a restriction is placed on the maximum length of paths in the cover (\cref{def:RLPC}).

\begin{definition}[\RLPCshort]\label{def:RLPC}
    Given an undirected graph $G$ and positive integers $k$ and $\ell$, the \RLPC\ (\RLPCshort) problem asks whether there exists a set of $k$ paths in $G$, each with length at most $\ell$, such that every vertex of $G$ belongs to at least one of the paths.
\end{definition}

\subsection{Hardness for constant deletion distance to linear forest}

We prove that \TPCshort\ and \TDPCshort\ remain \NP-hard even when $\ddlf(\underline{U_G}) = 1$,
through a reduction from \NMTS\ (\cref{def:NMTS})
.
See \cref{subsec:prelims_parameterized} for the definition and usefulness of the $\ddlf$ parameter.

\begin{definition}[\NMTSshort]\label{def:NMTS}
    Given three (multi)sets $A=\{a_1,\ldots,a_n\}$, $B=\{b_1,\ldots,b_n\}$ and $T=\{t_1,\ldots,t_n\}$ of positive integers, the \NMTS\ $(\NMTSshort)$ problem asks whether there is a subset $M$ of $A \times B \times T$ s.t. every $a_i \in A$, $b_i \in B$, $t_i \in T$ occurs exactly once in~$M$ and for every triplet $(a,b,t)\in M$ it holds that $a+b = t$.
\end{definition}

\begin{theorem}[Hulett, Will, Woeginger 2008~\cite{NMTS_distinct}]\label{thrm:NMTS_distinct}
    \NMTSshort\ is strongly \NP-complete even when all $3n$ input elements are distinct.
\end{theorem}

The following auxiliary lemma will be useful for our reductions.

\begin{lemmarep}[\appsymb]\label{lem:NMTS_hardness}
    \NMTSshort\ is strongly \NP-complete even when:
    \begin{enumerate}
        \item $\min(B)>2\max(A)$,
        \item $\min(T)>\max(B)$,
        \item all $3n$ input elements are distinct.
    \end{enumerate}
\end{lemmarep}

\begin{proof}
    Take an \NMTSshort\ instance $I=(A=\{a_1,\ldots,a_n\},B=\{b_1,\ldots,b_n\},T=\{t_1,\ldots,t_n\})$ and modify it as follows.
    \begin{itemize}
        \item Increase all $a_i\in A$ by $\max(B)$.
        \item Increase all $b_i\in B$ by $2\max(A)+2\max(B)$.
        \item Increase all $t_i\in T$ by $2\max(A)+3\max(B)$.
    \end{itemize}
    It is clear by definition that the resulting \NMTSshort\ instance is equivalent to $I$. Since $A,B,T$ contain (strictly) positive integers, it follows that the resulting instance $(A',B',T')$ satisfies $\min(B')>2\max(A')$ and $\min(T')>\max(B')$. The described transformation preserves strong \NP-hardness, since the values of all elements remain polynomial in $n$.

    By \cref{thrm:NMTS_distinct}, we can apply the above to an instance with distinct elements. It is clear that the resulting instance then also has distinct elements.
\end{proof}

We will first show hardness for \RLPC\ with constant $\ddlf(G)$. 
Our hardness construction is inspired by a recent result proving analogous hardness for the \textsc{Telephone Broadcast} problem~\cite{mfcs/EgamiGHKLMNOV025}. See \cref{fig:NMTS_reduction_main_body} for a sketch.

    \begin{figure}[ht]
        \centering
        \includegraphics[width=0.7\linewidth, page=3]{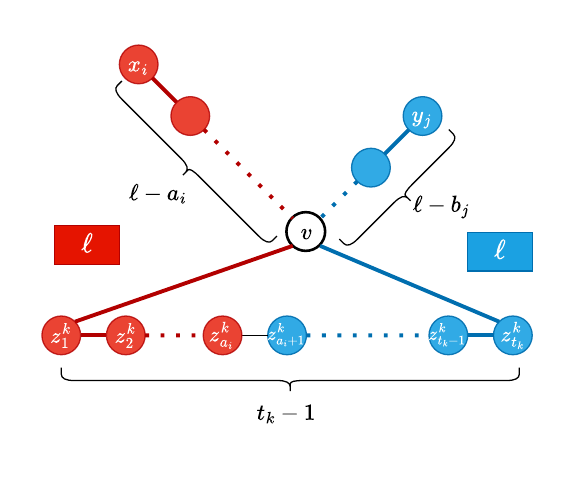}
        \caption{The reduction of \cref{thrm:restricted_length}, from \NMTSshort\ to \RLPCshort. The construction is only shown for a triplet $(a_i,b_j,t_k)$ with $a_i+b_j=t_k$ for simplicity. Paths of length $\ell$ start from the vertices $x_i,y_j$ and cover the red and blue vertices respectively. The same reduction can be applied to \TPCshort\ by assigning the label set $[\ell]$ to all edges, and to \TDPCshort\ through \cref{thrm:NMTS_distinct}, which essentially implies that $v$ is not occupied by two paths at the same time unit. See \cref{thrm:TPC_pathwidth} for more details.}
        \label{fig:NMTS_reduction_main_body}
    \end{figure}

\begin{theoremrep}[\appsymb]\label{thrm:restricted_length}
    \RLPCshort\ is \NP-hard even when $\ddlf(G)=1$.
\end{theoremrep}

\begin{proof}
    We will reduce \NMTSshort\ to \RLPCshort.

    By \cref{lem:NMTS_hardness}, we can assume an \NMTSshort\ instance $I=(A=\{a_1,\ldots,a_n\},B=\{b_1,\ldots,b_n\},T=\{t_1,\ldots,t_n\})$ with distinct elements and
    \begin{equation}\label{NMTS_eq1}
        \min(B)>2\max(A),\quad \min(T)>\max(B).
    \end{equation}
    Notice that for any non-trivial \NMTSshort\ instance it must be
    \begin{equation}\label{NMTS_eq2}
        \sum(A)+\sum(B)=\sum(T)
    \end{equation}
    From~\eqref{NMTS_eq1} it additionally follows that
    \begin{equation}\label{NMTS_eq3}
        2\max(A)<\min(T).
    \end{equation}
    We will later use the above inequalities for the correctness of the reduction.

    We now construct an \RLPCshort\ instance based on $I$.
    Define $\ell=2\max(B)+1$ and construct an undirected graph $G$ as follows. Create a vertex $v$ and $2n$ vertices $x_1,\ldots, x_n,y_1,\ldots,y_n$. Next, connect~$v$ to each $x_i$ ($i\in [n]$) through a path of length $\ell - a_i$ and $v$ to each $y_i$ ($i\in [n]$) through a path of length $\ell - b_i$. Lastly, for each $t_i$ ($i\in [n]$), create a path of length\footnote{By~\eqref{NMTS_eq3}, this length is strictly positive.} $t_i-1$ consisting of vertices $z_1^i,\ldots z_{t_i}^i$ and connect $z_1^i$ and $z_{t_i}^i$ to $v$. We refer to \cref{fig:NMTS_reduction} for a sketch of the described construction. As the removal of $v$ only leaves a disjoint union of paths, it is easy to see that $\ddlf(G)=1$.

    \begin{figure}[ht]
        \centering
        \includegraphics[width=0.75\linewidth, page=3]{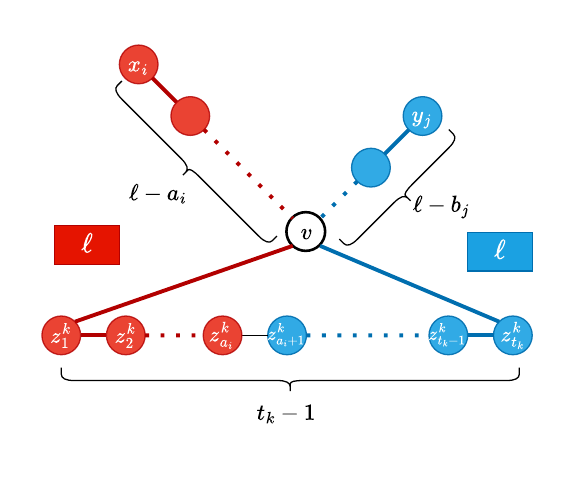}
        \caption{The reduction of \cref{thrm:restricted_length}, from \NMTSshort\ to \RLPCshort. The construction is only shown for a triplet $(a_i,b_j,t_k)$ with $a_i+b_j=t_k$ for simplicity. Paths of length $\ell$ start from the vertices $x_i,y_j$ and cover the red and blue vertices respectively. Equation~\eqref{NMTS_eq2} implies that the paths should not intersect, except at $v$. The same reduction can be applied to \TPCshort\ by assigning the label set $[\ell]$ to all edges, and to \TDPCshort\ through \cref{thrm:NMTS_distinct}, which essentially implies that $v$ is not occupied by two paths at any time unit. See \cref{thrm:TPC_pathwidth} for more details.}
        \label{fig:NMTS_reduction}
    \end{figure}

    We now prove the correctness of the reduction; specifically, that $I$ admits a solution if and only if $G$ has a path cover consisting of $2n$ paths, each of length at most $\ell=2\max(B)+1$.

    ($\Rightarrow$) Suppose that $I$ admits a solution, i.e., the elements in $A$, $B$ and $T$ can be matched in~$n$ triplets $(a_i,b_j,t_k)\in A\times B\times T$ such that $a_i+b_j=t_k$. For each such triplet, pick two paths of length $\ell$ each, starting from $x_i$ and $y_j$. The former path covers all vertices in the path from $x_i$ to $v$ and the latter covers all vertices in the path from $y_j$ to $v$; the former can then cover at most $a_i$ other vertices and the latter at most $b_j$ other vertices. Thus, we can have the former cover the vertices $z_1^k,\ldots,z_{a_i}^k$ and the latter cover the vertices $z_{t_k}^k,\ldots,z_{t_k-b_j+1}^k$. Notice that the equality $a_i+b_j=t_k$ immediately implies that all vertices in $G$ are covered, since $z_{t_k-b_j+1}^k \equiv z_{a_i+1}^k$.

    ($\Leftarrow$) Suppose $G$ has a path cover of $2n$ paths, each of length at most $\ell$. By construction, the distance between any two vertices in $\{x_1,\ldots, x_n,y_1,\ldots,y_n\}$ is at least
    \[
        2(\ell - \max(B))=2\max(B)+2>\ell.
    \]

    Thus, it is impossible for any two vertices in $\{x_1,\ldots, x_n,y_1,\ldots,y_n\}$ to be covered by one path of length $\ell$. It follows that the $2n$ paths in the path cover have to start from the vertices $x_1,\ldots, x_n,y_1,\ldots,y_n$. Similar to the forward direction of the correctness proof, for $i\in [n],\ j\in [n]$, the path starting from $x_i$ can cover $a_i$ $z$-vertices and the path starting from $y_j$ can cover $b_j$ $z$-vertices. Combining this with~\eqref{NMTS_eq2}, it follows that each $z$-vertex can only be covered by a single path (otherwise, some $z$-vertex in $G$ would remain uncovered).
    Additionally, \eqref{NMTS_eq1} implies that, for all $k\in [n]$, the vertices $z_1^k,\ldots z_{t_k}^k$ cannot all be covered by the same path; thus, they have to be collectively covered by two paths (they cannot be covered by three or more paths, since each $z$-vertex can only be covered by one path).
    Due to~\eqref{NMTS_eq3}, these two paths cannot both start from $x$-vertices; thus, one has to start from an $x$-vertex and the other from a $y$-vertex.\footnote{If for some $k\in [n]$ both paths started from a $y$-vertex, this would imply that there exists a $k'\in[n]$, $k'\neq k$, for which the two paths covering the vertices $z_1^{k'},\ldots z_{t_{k'}}^{k'}$ both start from $x$-vertices, which is a contradiction.} We infer that, for all $k\in [n]$, the vertices $z_1^k,\ldots z_{t_k}^k$ are covered by paths starting from $x_i$ and $y_j$ such that $a_i+b_j=t_k$, which immediately implies that $I$ admits a solution (by matching the respective triplets $(a_i,b_j,t_k)$).

    Since \NMTSshort\ is \emph{strongly} \NP-complete, we can assume that $\ell=2\max(B)+1$ is polynomial in $n$. It follows that the described reduction runs in polynomial time. 
\end{proof}

\cref{thrm:restricted_length} implies that \RLPC\ is \NP-hard on graphs of constant treewidth,
in contrast to \PC, which is known to admit an \XP\ algorithm parameterized by treewidth~\cite{Foucaud_static_caldam}.

We next adapt the reduction of \cref{thrm:restricted_length} to \TPCshort\ and \TDPCshort.

\begin{theoremrep}[\appsymb]\label{thrm:TPC_pathwidth}
    \TPCshort\ and \TDPCshort\ are \NP-hard even when $\ddlf(\underline{U_G})=1$ and the label set $\lambda$ is identical for all edges.
\end{theoremrep}

\begin{proof}
    The reduction of \cref{thrm:restricted_length} can be applied to both \TPCshort\ and \TDPCshort\ with the following minor modification.
    For \TPCshort, it suffices to assign to each edge the label set $[\ell]$. It is immediate that this causes the length of any temporal path to be bounded by $\ell$, and all paths in $G$ with length up to $\ell$ correspond to a feasible temporal path in the described temporal graph.
    The same reduction also works for \TDPCshort, since the $2n$ paths starting from the vertices $x_1,\ldots, x_n,y_1,\ldots,y_n$ all occupy $v$ in different time units, due to the fact that $a_1,\ldots, a_n,b_1,\ldots,b_n$ are distinct numbers by \cref{lem:NMTS_hardness}.
\end{proof}

\cref{thrm:TPC_pathwidth} complements the \NP-hardness of \TPCshort\ and \TDPCshort\ from~\cite{temporal_path_cover_Dailly_MFCS}, showing hardness even on highly structured graphs.
Note that \TDPCshort\ is already known to be \NP-hard on \emph{temporal oriented trees}~\cite{temporal_path_cover_Dailly_MFCS}, while a closer look into the proof reveals that for the constructed instance it holds that it has vertex-deletion distance $1$ to a forest whose every connected component is a star.
Our hardness result thus complements this by showing that the problem remains \NP-hard even when the graph is one vertex-deletion away from yet another highly restricted graph class,
this time under the additional restriction that all edges use the same label set.
As for \TPCshort, prior to our work it was not known whether it remains \NP-hard even for graphs of constant treewidth. Moreover, since $\ddlf(\underline{U_G})=1$ implies $\fvs(\underline{U_G})\leq 1$, \cref{thrm:TPC_pathwidth} highlights an interesting contrast: \TPCshort\ admits a polynomial-time algorithm on temporal oriented trees~\cite{temporal_path_cover_Dailly_MFCS}, i.e., when $\fvs(\underline{U_G})=0$, but becomes \NP-hard when $\fvs(\underline{U_G})=1$.

\subsection{Hardness for constant maximum degree and pathwidth}

To complement the aforementioned results, we present a slight variation of the previous reduction in order to prove that \RLPCshort\ is \NP-hard even on bounded-degree instances. The main idea is to replace the vertex $v$ in the graph of \cref{fig:NMTS_reduction_main_body} with a cycle, in order to force every vertex in the graph to have constant degree.

\begin{theoremrep}[\appsymb]\label{thrm:restricted_length2}
    \RLPCshort\ is \NP-hard even in graphs with $\Delta(G)=3$ and
    $\pw(G)=3$.
\end{theoremrep}

\begin{proof}
    Similarly to the proof of \cref{thrm:restricted_length}, we will reduce an \NMTSshort\ instance $I=(A=\{a_1,\ldots,a_n\},B=\{b_1,\ldots,b_n\},T=\{t_1,\ldots,t_n\})$ with distinct elements and satisfying properties \eqref{NMTS_eq1}, \eqref{NMTS_eq2} and \eqref{NMTS_eq3}.

    We construct an $\RLPCshort$ instance based on $I$. Define $\ell=2\max(B)+1$ and construct an undirected graph $G$ as follows. Create a cycle $C$ of length $4n$ $v_1,\dots,v_{4n}$ and $2n$ vertices $x_1,\dots,x_n,y_1,\dots,y_n$. Next, connect each $x_i$ to $v_i$ through a path of length $(4n+1)(\ell-a_i)$, then connect each $y_j$ to $v_{n+j}$ through a path of length $(4n+1)(\ell-b_j)$. Lastly, for each $t_k$ $(k\in [n])$, create a path of length $(4n+1)(t_k+1)-2$ consisting of vertices $z_1^k,\ldots, z_{(4n+1)(t_k+1)-1}^k$ and connect $z_1^k$ and $z_{(4n+1)(t_k+1)-1}^k$ to $v_{2n+2k-1}$ and $v_{2n+2k}$ respectively. Note that all vertices of $C$ have degree $3$, and all other vertices have degree $2$ or $1$. It is straightforward to prove that $\pw(G)=3$.

    We now prove the correctness of the reduction in a very similar fashion as in the proof of \cref{thrm:restricted_length}. We prove that $I$ admits a solution if and only if $G$ has a path cover consisting of $2n$ paths, each of length at most $(4n+1)\ell+4n$.

    ($\Rightarrow$) Suppose that $I$ admits a solution, i.e., the elements of $A$, $B$ and $T$ can be matched in $n$ triplets $(a_i,b_j,t_k)\in A\times B\times T$ such that $a_i+b_j=t_k$. For each such triplet, pick two paths starting from $x_i$ and $y_j$ of size at most $(4n+1)\ell+4n$. The former path covers all vertices in the path from $x_i$ to $v_i$ and the latter covers all vertices in the path from $y_j$ to $v_{n+j}$; the distances between $v_i$ and $v_{2n+2k-1}$, and $v_{n+j}$ and $v_{2n+2k}$ (through $C$) are both bounded by $|C|/2=2n$. Therefore, going through $C$, the former path can cover vertices $z_1^k,\dots,z_{(4n+1)a_i+2n}^k$, while the latter path can cover vertices $z_{(4n+1)(t_k+1)-1}^k,\dots,z_{(4n+1)(t_k-b_j)+2n+1}^k$.
    Notice that the equality $a_i+b_j=t_k$ implies that all vertices in $G$ are covered, since $z_{(4n+1)(t_k-b_j)+2n}^k\equiv z_{(4n+1)a_i +2n}^k$.

    ($\Leftarrow$) Suppose $G$ has a path cover of $2n$ paths, each of length at most $(4n+1)\ell+4n$. By construction , the distance between any two vertices in $\{x_1,\dots,x_n,y_1,\dots,y_n\}$ is at least
    \[
        2(4n+1)(\ell - \max(B))
        =(4n+1)(2\max(B)+2)>(4n+1)\ell+4n.
    \]

    Thus, it is impossible for any two vertices in $\{x_1,\dots,x_n,y_1,\dots,y_n\}$ to be covered by one path of length $(4n+1)\ell+4n$. It follows that the $2n$ paths in the path cover have to start from the vertices $x_1,\dots,x_n,y_1,\dots,y_n$. Similar to the forward direction of the correctness proof, for $i\in [n]$, $j\in [n]$, the path starting from $x_i$ can cover between $(4n+1)a_i+2n$ and $(4n+1)a_i+4n$ $z$-vertices and the path starting from $y_j$ can cover between $(4n+1)b_j+2n$ and $(4n+1)b_j+4n$ $z$-vertices.

    Therefore, due to~\eqref{NMTS_eq1}, for $k\in [n]$, at least two paths are needed to cover $z_1^k,\dots,z_{(4n+1)t_k}^k$. Furthermore, for $k\in[n]$, a path cannot cover some vertices of $z_1^k,\dots,z_{(4n+1)(t_k+1)-1}^k$ and, for $k'\in [n]$, $k\neq k'$, some vertices of $z_1^{k'},\dots,z_{(4n+1)(t_{k'}+1)-1}^{k'}$, as it would have to enter $z_1^k,\dots,z_{(4n+1)(t_k+1)-1}^k$ through either $z_1^k$ or $z_{(4n+1)(t_k+1)-1}^k$, and leave from the other, which is impossible due to $\eqref{NMTS_eq1}$. Hence, as there are only $2n$ paths, for every $k\in [n]$, $z_1^k,\dots,z_{(4n+1)(t_k+1)-1}^k$ are covered by exactly two paths. Due to~\eqref{NMTS_eq3}, these two paths cannot both start from $x$-vertices; thus, one has to start from an $x$-vertex and the other from a $y$-vertex. Let $x_i$ and $y_j$ be those vertices. We have:
    \[
        (4n+1)a_i+4n+(4n+1)b_j+4n
        \geq (4n+1)(t_k+1)-1
        \iff a_i+b_j+\frac{4n}{4n+1}\geq t_k.
    \]
    Since $a_i,b_j,t_k$ are integers, this implies
    \[
        a_i+b_j\geq t_k.
    \]
    Since this holds for every $k\in [n]$, we obtain the following in conjunction with~\eqref{NMTS_eq2}.
    \[
        a_i+b_j=t_k,
    \]
    which implies that $I$ admits a solution by matching the respective triplets.
\end{proof}

Lastly, we adapt \cref{thrm:restricted_length2} to TPC and TDPC. The proof is almost identical to that of \cref{thrm:TPC_pathwidth}, and is thus omitted.

\begin{theorem}\label{thrm:TPC_pathwidth2}
    \TPCshort\ and \TDPCshort\ are \NP-hard even when $\Delta(\underline{U_G})=3$ and $\pw(\underline{U_G})=3$ and the label set $\lambda$ is identical for
    all edges.
\end{theorem}

\subsection{Hardness for constant vertex cover number}

By \cref{thrm:TPC_vc_tau}, \TPCshort\ is \FPT\ parameterized by
$\vc(\underline{U_G})+t_{\max}$; the same holds for \TDPCshort\
by~\cite{temporal_path_cover_Dailly_MFCS}. We now show that the dependence on
$t_{\max}$ is necessary: both problems are already \NP-hard when
$\vc(\underline{U_G})=3$.

We reduce from \NMTSshort. Recall that, given three sets
$A=\{a_1,\ldots,a_n\}$, $B=\{b_1,\ldots,b_n\}$, and
$T=\{t_1,\ldots,t_n\}$, the problem asks whether their elements can be
partitioned into triples $(a_i,b_j,t_k)$ satisfying $a_i+b_j=t_k$.
Our construction is organised around three vertices $x_1,x_2,x_3$.
For each $a_i$ we introduce a vertex $\alpha_i$ and edges $x_1 \to \alpha_i \to x_2$,
while for each $t_k$ we introduce a vertex $\beta_k$ and edges $x_2 \to \beta_k \to x_3$.
Lastly, for each $b_j$, we introduce a source $p_j$ and a sink $q_j$, along with edges $p_j \to x_1$ and $x_3 \to q_j$ respectively.
Setting the time-slots appropriately forces the path that starts from $p_j$ to end in $q_j$,
thus the path
\[
    p_j\to x_1\to\alpha_i\to x_2\to\beta_k\to x_3\to q_j
\]
selects $(a_i,b_j,t_k)$ and is a valid temporal path exactly when
$a_i+b_j\leq t_k$. A cover of size $n$ uses every element of $A,B,T$ exactly
once; since $\sum A+\sum B=\sum T$, all selected inequalities must become
equalities. Moreover, the paths are temporally disjoint, so the same construction applies to \TDPCshort.
We refer to \cref{fig:vc_reduction} for a sketch of the construction.

\begin{figure}[ht]
    \centering
    \begin{tikzpicture}[
        >=stealth, xscale=1.1, yscale=1.0,
        cover/.style={circle, draw, thick, fill=orange!30,
                      inner sep=0pt, minimum size=7mm, font=\small},
        src/.style={circle, draw, semithick, fill=white,
                    inner sep=0pt, minimum size=6.5mm, font=\small},
        snk/.style={circle, draw, semithick, fill=black!12,
                    inner sep=0pt, minimum size=6.5mm, font=\small},
        av/.style={circle, draw, semithick, fill=cyan!20,
                   inner sep=0pt, minimum size=6.5mm, font=\small},
        tv/.style={circle, draw, semithick, fill=green!25,
                   inner sep=0pt, minimum size=6.5mm, font=\small},
        fan/.style={->, gray!55, thin},
        spine/.style={->, line width=1.1pt, red!75!black},
        tlab/.style={font=\scriptsize, inner sep=1.2pt, fill=white},
        layer/.style={font=\footnotesize},
    ]
        \def\yT{1.4}\def\yB{-1.4}
        \node[src] (p1) at (0,\yT) {$p_1$};
        \node[src] (pj) at (0,0)   {$p_j$};
        \node[src] (pn) at (0,\yB) {$p_n$};

        \node[cover] (x1) at (2,0) {$x_1$};

        \node[av] (a1) at (4,\yT) {$\alpha_1$};
        \node[av] (ai) at (4,0)   {$\alpha_i$};
        \node[av] (an) at (4,\yB) {$\alpha_n$};

        \node[cover] (x2) at (6,0) {$x_2$};

        \node[tv] (b1) at (8,\yT) {$\beta_1$};
        \node[tv] (bk) at (8,0)   {$\beta_k$};
        \node[tv] (bn) at (8,\yB) {$\beta_n$};

        \node[cover] (x3) at (10,0) {$x_3$};

        \node[snk] (q1) at (12,\yT) {$q_1$};
        \node[snk] (qj) at (12,0)   {$q_j$};
        \node[snk] (qn) at (12,\yB) {$q_n$};

        \foreach \x in {0,4,8,12}{
            \node at (\x, 0.7)  {$\vdots$};
            \node at (\x,-0.7)  {$\vdots$};
        }

        \foreach \n in {p1,pn}{\draw[fan] (\n) -- (x1);}
        \foreach \n in {a1,an}{\draw[fan] (x1) -- (\n); \draw[fan] (\n) -- (x2);}
        \foreach \n in {b1,bn}{\draw[fan] (x2) -- (\n); \draw[fan] (\n) -- (x3);}
        \foreach \n in {q1,qn}{\draw[fan] (x3) -- (\n);}

        \draw[spine] (pj) -- node[tlab]{$1$}            (x1);
        \draw[spine] (x1) -- node[tlab]{$2$}            (ai);
        \draw[spine] (ai) -- node[tlab]{$a_i{+}b_j{+}2$}(x2);
        \draw[spine] (x2) -- node[tlab]{$t_k{+}3$}      (bk);
        \draw[spine] (bk) -- node[tlab]{$M{+}4$}        (x3);
        \draw[spine] (x3) -- node[tlab]{$M{+}5$}        (qj);

        \node[layer] at (0,   -2.05) {$B$};
        \node[layer] at (4, -2.05) {$A$};
        \node[layer] at (8, -2.05) {$T$};
        \node[layer] at (12,-2.05) {$B$};
    \end{tikzpicture}
    \caption{The reduction of \cref{thrm:TPC_vc_hardness}. The highlighted
        path $P_{i,j,k}$ selects $(a_i,b_j,t_k)$; the displayed labels are its
        offsets within block~$j$, so it is a valid temporal path exactly when
        $a_i+b_j\leq t_k$. The set $\{x_1,x_2,x_3\}$ is a vertex cover.}
    \label{fig:vc_reduction}
\end{figure}

\begin{theoremrep}[\appsymb]\label{thrm:TPC_vc_hardness}
    \TPCshort\ and \TDPCshort\ are \NP-hard even on temporal DAGs with $\vc(\underline{U_G}) = 3$.
\end{theoremrep}

\begin{proof}
    We reduce from \NMTSshort, which is strongly \NP-hard by
    \cref{thrm:NMTS_distinct}. Let
    \[
        A=\{a_1,\ldots,a_n\},\qquad
        B=\{b_1,\ldots,b_n\},\qquad
        T=\{t_1,\ldots,t_n\}
    \]
    be an instance of \NMTSshort. We may assume that
    \begin{equation}\label{eq:vc_total_sum}
        \sum_{i\in[n]}a_i+\sum_{j\in[n]}b_j
        =\sum_{k\in[n]}t_k,
    \end{equation}
    as well as that $\max_{i,j\in[n]}\{a_i+b_j\} \ge \max_{k\in[n]} t_k$,
    as otherwise this is a trivial no-instance.
    Set $M \coloneq \max_{i,j\in[n]}\{a_i+b_j\}$.

    \proofsubparagraph{Construction.}
    Introduce three vertices $x_1,x_2,x_3$. For every $a_i\in A$ and
    $t_k\in T$, introduce vertices $\alpha_i$ and $\beta_k$, respectively.
    For every $b_j\in B$, introduce a source $p_j$ and a sink $q_j$.
    For $j\in[n]$, let $L_j \coloneq (j-1)(M+6)$.
    The arcs and their label sets are
    \[
    \begin{array}{rclcrcl}
        \lambda(p_j,x_1)       &=& \{L_j+1\},
        &\qquad&
        \lambda(x_1,\alpha_i)  &=& \setdef{L_h+2}{h\in[n]},\\[1ex]
        \lambda(\alpha_i,x_2)  &=&
            \setdef{L_h+a_i+b_h+2}{h\in[n]},
        &&
        \lambda(x_2,\beta_k)  &=&
            \setdef{L_h+t_k+3}{h\in[n]},\\[1ex]
        \lambda(\beta_k,x_3)  &=&
            \setdef{L_h+M+4}{h\in[n]},
        &&
        \lambda(x_3,q_j)       &=& \{L_j+M+5\}.
    \end{array}
    \]
    Set the target cover size to $K \coloneq n$.

    For $i,j,k\in[n]$, let $P_{i,j,k}$ be the path
    \[
        p_j \xrightarrow{L_j+1} x_1 \xrightarrow{L_j+2} \alpha_i \xrightarrow{L_j+a_i+b_j+2} x_2 \xrightarrow{L_j+t_k+3} \beta_k \xrightarrow{L_j+M+4} x_3 \xrightarrow{L_j+M+5} q_j,
    \]
    where we prescribe the time-slot an edge is used.
    By the choice of $M$, it holds that
    \begin{equation}\label{eq:vc_path_inequality}
        P_{i,j,k}\text{ is temporal}
        \quad\Longleftrightarrow\quad
        L_j+a_i+b_j+2<L_j+t_k+3
        \quad\Longleftrightarrow\quad
        a_i+b_j\leq t_k.
    \end{equation}

    \begin{lemma}\label{lem:vc_NMTS_to_cover}
        If the \NMTSshort\ instance is positive, then $G$ has a temporally
        disjoint path cover of size $n$.
    \end{lemma}

    \begin{nestedproof}
        Let $\mathcal{M}$ be a solution of the \NMTSshort\ instance. For every
        $(a_i,b_j,t_k)\in\mathcal{M}$, select $P_{i,j,k}$. By
        \cref{eq:vc_path_inequality}, each selected path is temporal, and the
        paths cover every vertex because $\mathcal{M}$ uses every element of
        $A,B,T$ exactly once.

        Moreover, the path associated with $b_j$ occupies vertices only
        during $[L_j+1,L_j+M+6]$. These intervals are pairwise disjoint, so
        the selected paths are temporally disjoint.
    \end{nestedproof}

    \begin{lemma}\label{lem:vc_cover_to_NMTS}
        If $G$ has a temporal path cover of size at most $n$, then the
        \NMTSshort\ instance is positive.
    \end{lemma}

    \begin{nestedproof}
        Let $\mathcal{C}$ be such a cover. The vertices
        $p_1,\ldots,p_n$ have in-degree zero, and no path contains two of
        them. Hence every cover has at least $n$ paths. Since
        $|\mathcal{C}|\leq n$, it consists of exactly $n$ paths, one starting
        at each $p_j$. Similarly, no path contains two of the sinks
        $q_1,\ldots,q_n$. Since all $n$ sinks must be covered, every path
        contains exactly one of them.

        \begin{claim}\label{clm:vc_block_confinement}
            For every $j\in[n]$, the path starting at $p_j$ ends at $q_j$ and
            uses only labels in the interval $[L_j+1,L_j+M+5]$.
        \end{claim}

        \begin{claimproof}
            Let the path starting at $p_j$ end at $q_{\rho(j)}$. Since every
            sink occurs on exactly one path, $\rho$ is a permutation of
            $[n]$. If $\rho(j)=h<j$, then its last label is
            \[
                L_h+M+5<L_j+1,
            \]
            which precedes its first label. Hence $\rho(j)\geq j$ for every
            $j$, and therefore $\rho$ is the identity.
            Since the path from $p_j$ to $q_j$ starts at time $L_j+1$ and ends at
            time $L_j+M+5$, the claim follows.
        \end{claimproof}

        By the layered structure and \cref{clm:vc_block_confinement}, the path
        starting at $p_j$ is $P_{i,j,k}$ for some $i,k\in[n]$. Hence
        \cref{eq:vc_path_inequality} gives $a_i+b_j\leq t_k$.
        Thus every element of $B$ is used once. Moreover, the $n$ paths each
        contain one $\alpha$-vertex and one $\beta$-vertex. Since all such
        vertices are covered, they also use every element of $A$ and $T$
        exactly once. Summing the resulting inequalities and applying
        \cref{eq:vc_total_sum} shows that all of them hold with equality.
        The corresponding triples form an \NMTSshort\ solution.
    \end{nestedproof}

    The two lemmas establish correctness for both problems: the forward
    direction constructs a temporally disjoint cover, while the reverse
    direction applies to every temporal path cover.
    The construction has $4n+3$ vertices, $6n$ arcs, and $O(n^2)$ labels.
    Since \NMTSshort\ is strongly \NP-hard, all labels are polynomially
    bounded, and the reduction is polynomial.
    It is straightforward to see that the constructed graph is a DAG.
    Finally, $\vc(\underline{U_G})\leq3$ as $\{x_1,x_2,x_3\}$ is a vertex cover.
    This completes the proof.
\end{proof}

\section{Conclusion}

Our work paints a complete picture of the relationship between \TPCshort, \TDPCshort, \TACshort, and the Dilworth and TD-Dilworth properties. We establish a plethora of intriguing contrasts for these problems: the respective Dilworth promises guarantee tractability for the decision versions of \TPCshort\ and \TDPCshort, however, this tractability does not hold for the search versions of the problems. In addition, when the (TD-)Dilworth promise is slightly relaxed to guarantee only that the difference between the minimum path cover size and maximum antichain is at most $1$ (instead of $0$), even the decision versions of both problems become \NP-hard. An interesting question arising from this is whether a multiplicative (instead of additive) relaxation of the (TD-)Dilworth promise would guarantee some form of tractability for the problems, perhaps by enforcing an approximate sandwich theorem for the Lov{\'a}sz number of the connectivity graph through \cref{lem:lovasz_inequality_chain}.

We additionally answer a question by Chakraborty et al.~\cite{temporal_path_cover_Dailly_MFCS} by proving that \TPCshort\ parameterized by $\tw + t_{\max}$ does not admit an \FPT\ algorithm under standard parameterized complexity assumptions, contrasting the status of \TDPCshort, for which such an algorithm is known~\cite{temporal_path_cover_Dailly_MFCS}. We remark that a similar question for the static \PC\ problem remains open; that is, is the \XP\ algorithm of Foucaud et al.~\cite{Foucaud_static_caldam} parameterized by treewidth optimal or does the problem admit an \FPT\ one?
Moreover, can we obtain conditional lower bounds for \TPCshort\ and \TDPCshort\ with the parameterization by $\tw + t_{\max}$, or improve the running times of the respective algorithms from~\cite{temporal_path_cover_Dailly_MFCS}?
On the parameterized front, another interesting question is whether \cref{thrm:TPC_vc_hardness} is tight, that is, whether \TPCshort\ and \TDPCshort\ remain \NP-hard even when $\vc(\underline{U_G}) \le 2$.\footnote{We remark that a very recent preprint~\cite{tight_vc} claims to have answered this question, proving \NP-hardness for both problems even when $\vc(\underline{U_G}) = 2$.}
Lastly, we remark that although parameterized complexity and structural graph theory provide us with the tools to describe the structure of the input graph,
they inherently fail to capture the particularities of \emph{temporal} graphs. It would be thus interesting to explore temporal analogues of them,
where for instance, in the case of treewidth, the tree-decomposition takes into account the labels of the edges in a less coarse sense (compared to mere existence in the underlying graph).

\subsubsection*{Acknowledgements}
The authors are grateful to the organizers of IWOCA 2026, held at Université Clermont Auvergne. In particular, we would like to thank Antoine Dailly and Florent Foucaud for suggesting the open questions tackled here in the open problem session of the collaborative workshop of IWOCA 2026, as well as Anni Hakanen for her contribution to the discussion.

The authors would like to thank Argyrios Deligkas and Eduard Eiben for valuable discussions regarding temporal path covers. Additionally, the authors are grateful to the anonymous reviewer in ISAAC 2026 who suggested the proof of \cref{thrm:TPC_search_dilworth}.

S.K. and C.P. gratefully acknowledge the financial support provided through the IWOCA
2026 Extended Stay Support Scheme, which enabled a research visit to Université Paris
Dauphine, where part of this work was carried out. The authors also thank Laurent Gourvès
and Michael Lampis for organizing the visit and facilitating the financial support that made
it possible.

L.C. was partially funded by Project 58523\_DIPECC\_23\_27 - ``Finanziamento Dipartimenti di Eccellenza 2023-2027'' - B13C22004500001.

L.C. was partially funded by MUR of Italy, under PRIN Project n. 2022ME9Z78 - NextGRAAL: Next-generation algorithms for constrained GRAph visuALization.

S.K., A.P., and C.P. were partially supported by project MIS 5154714 of the National Recovery and Resilience Plan Greece 2.0 funded by the European Union under the NextGenerationEU Program.

S.K. and C.P. acknowledge the financial support of the Archimedes Research Unit (Athena RC) during their PhD.

E.N. and M.V. were partially supported by the ANR project ANR-21-CE48-0022 (S-EX-AP-PE-AL).

S.K., A.P., C.P., and M.V. were partially supported by the Partenariat Hubert Curien France – Grèce 2025 programme ``Aristote'', funded by the Ministry of Europe and Foreign Affairs (MEAE) and the Ministry of Higher Education, Research, and Innovation (MESRI) of France, and the State Scholarships Foundation (IKY) of Greece, via the Cooperation and Cultural Action Service of the French Embassy, and the French Institute of Greece.

S.K. and M.V. were partially supported by the project COAL-GAS, financed under the
Global Seed Fund program of Université Paris Sciences et Lettres (PSL) and by the CNRS.

M.V. was partially supported by the ERC grant PARAMLP (project no. 101220563) funded by the European Union.
Views and opinions expressed are however those of the author(s) only and do not necessarily reflect those of the European Union or the European Research Council Executive Agency. Neither the European Union nor the granting authority can be held responsible for them.



\bibliography{bibliography}

\end{document}